\documentclass[final,12pt]{article}
\usepackage{hyperref}
\usepackage{amssymb,amsmath,amsthm}
\usepackage{tikz}
\usetikzlibrary{arrows,backgrounds,decorations.pathmorphing,decorations.pathreplacing,positioning,fit}
\usepackage{enumerate}
\usepackage[conditional,light,first,bottomafter]{draftcopy}
\draftcopyName{DRAFT\space\today}{130}
\draftcopySetScale{65}
\usepackage[letterpaper,hmargin=3.3cm,vmargin={2.6cm,3.2cm}]{geometry}
\usepackage{setspace}
\makeatletter
\renewcommand{\section}{\@startsection%
{section}%
{1}%
{0em}%
{1.7em}%
{1.2em}%
{\normalfont\large\centering\bfseries}}
\renewcommand{\@seccntformat}[1]%
{\csname the#1\endcsname.\hspace{0.5em}}
\makeatother

\numberwithin{equation}{section}
\newtheorem{theorem}{Theorem}[section]
\newtheorem{proposition}{Proposition}[section]
\newtheorem{lemma}{Lemma}[section]

\theoremstyle{definition}

\newtheorem{remark}{Remark}

\newcommand{\abs}[1]{\left|#1\right|}

\newcommand{\inner}[2]{\left\langle#1,#2\right\rangle}
\DeclareMathOperator{\im}{Im}
\DeclareMathOperator{\dom}{dom}
\DeclareMathOperator*{\res}{Res}
\begin{document}
\begin{titlepage}
\title{Inverse problems for Jacobi operators III:\\
Mass-spring perturbations of semi-infinite systems
\footnotetext{%
Mathematics Subject Classification(2010):
34K29,  % Inverse problems (34-XX Ordinary differential equations 34Kxx
       % Functional-differential and differential-difference
       % equations)
47A75, % Eigenvalue problems (47-XX Operator theory 47Axx General
       % theory of linear operators)
47B36, % Jacobi (tridiagonal) operators (matrices) and generalizations
       % (47-XX Operator theory 47Bxx Special classes of linear operators)
70F17, % Inverse problems (70-XX Mechanics of particles and systems
       % 70Fxx Dynamics of a system of particles, including celestial mechanics)
}
\footnotetext{%
Keywords:
Infinite mass-spring system;
Jacobi matrices;
Two-spectra inverse problem.
}
}
\author{
\textbf{Rafael del Rio}
\\
%% ----- Institution ________
\small Departamento de M\'{e}todos Matem\'{a}ticos y Num\'{e}ricos\\[-1.6mm]
\small Instituto de Investigaciones en Matem\'aticas Aplicadas y en Sistemas\\[-1.6mm]
\small Universidad Nacional Aut\'onoma de M\'exico\\[-1.6mm]
\small C.P. 04510, M\'exico D.F.\\[-1.6mm]
\small\texttt{delrio@leibniz.iimas.unam.mx}
\\[4mm]
\textbf{Mikhail Kudryavtsev}
\\
\small Department of Mathematics\\[-1.6mm]
\small Institute for Low Temperature Physics and Engineering\\[-1.6mm]
\small Lenin Av. 47, 61103\\[-1.6mm]
\small Kharkov, Ukraine\\[-1.6mm]
\small\texttt{kudryavtsev@onet.com.ua}
\\[4mm]
\textbf{Luis O. Silva}\thanks{%
Partially supported by CONACYT (M\'exico) through grant CB-2008-01-99100
}%
\\
%% ----- Institution ________
\small Departamento de M\'{e}todos Matem\'{a}ticos y Num\'{e}ricos\\[-1.6mm]
\small Instituto de Investigaciones en Matem\'aticas Aplicadas y en Sistemas\\[-1.6mm]
\small Universidad Nacional Aut\'onoma de M\'exico\\[-1.6mm]
\small C.P. 04510, M\'exico D.F.\\[-1.6mm]
\small \texttt{silva@leibniz.iimas.unam.mx}
}
%%%%%%%%
\date{}
\maketitle
\vspace{-4mm}
\begin{center}
\begin{minipage}{5in}
  \centerline{{\bf Abstract}}
\bigskip
Consider an infinite linear mass-spring system and a modification of it
obtained by changing the first mass and spring of the system. We give results
on the interplay of the spectra of such systems and on the
reconstruction of the system from its spectrum and the one of the
modified system. Furthermore, we provide necessary and sufficient
conditions for two sequences to be the spectra of the mass-spring
system and the perturbed one.
\end{minipage}
\end{center}
\thispagestyle{empty}
\end{titlepage}
%%%%%%%%%%%%%%%%%%%%%%%%%%%%%%
\section{Introduction}
\label{sec:intro}
 Let $l_{\rm fin}(\mathbb{N})$ be the linear space of complex sequences with a
finite number of non-zero elements. In the Hilbert space
$l_2(\mathbb{N})$, consider the operator $J_0$ defined for every
$f=\{f_k\}_{k=1}^\infty$ in $l_{\rm fin}(\mathbb{N})$ by
\begin{align}
  \label{eq:initial-spectral}
  (J_0f)_1&:= q_1 f_1 + b_1 f_2\,,\\
  \label{eq:recurrence-spectral}
  (J_0f)_k&:= b_{k-1}f_{k-1} + q_k f_k + b_kf_{k+1}\,,
  \quad k \in \mathbb{N} \setminus \{1\},
\end{align}
where $q_n\in\mathbb{R}$ and $b_n>0$ for any $n\in\mathbb{N}$. The operator
$J_0$ is symmetric and has deficiency indices $(1,1)$ or $(0,0)$
\cite[Chap.\,4,\,Sec.\,1.2]{MR0184042}. Fix a self-adjoint extension
of $J_0$ and denote it by $J$. Thus, either $J\varsupsetneq \overline{J_0}$
or $J=\overline{J_0}$. According to the definition of the matrix
representation for an unbounded symmetric operator
\cite[Sec. 47]{MR1255973}, $\overline{J_0}$ is the operator whose matrix
  representation with respect to the canonical basis
  $\{\delta_n\}_{n=1}^\infty$ in $l_2(\mathbb{N})$ is
\begin{equation}
  \label{eq:jm-0}
  \begin{pmatrix}
    q_1 & b_1 & 0  &  0  &  \cdots
\\[1mm] b_1 & q_2 & b_2 & 0 & \cdots \\[1mm]  0  &  b_2  & q_3  &
b_3 &  \\
0 & 0 & b_3 & q_4 & \ddots\\ \vdots & \vdots &  & \ddots
& \ddots
  \end{pmatrix}\,.
\end{equation}

Along with $J$, we consider the operator
\begin{equation}
\label{eq:def-tilde-j}
\begin{split}
  \widetilde{J}=J &+
  [q_1(\theta^2-1)+\theta^2h]\inner{\delta_1}{\cdot}\delta_1 \\
  &+  b_1(\theta-1)(\inner{\delta_1}{\cdot}\delta_2 +
  \inner{\delta_2}{\cdot}\delta_1)\,,\quad \theta>0\,,
  \quad h\in\mathbb{R}\,,
\end{split}
\end{equation}
which is a self-adjoint extension of the operator whose matrix
representation with respect to the canonical basis in
$l_2(\mathbb{N})$ is
\begin{equation}
  \label{eq:jm-1}
  \begin{pmatrix}
    \theta^2(q_1+h) & \theta b_1 & 0  &  0  &  \cdots
\\[1mm] \theta b_1 & q_2 & b_2 & 0 & \cdots \\[1mm]  0  &  b_2  & q_3  &
b_3 &  \\
0 & 0 & b_3 & q_4 & \ddots\\ \vdots & \vdots &  & \ddots
& \ddots
  \end{pmatrix}\,.
\end{equation}
Note that $\widetilde{J}$ is obtained from $J$ by a particular kind of
rank-two perturbation.

Under the assumption that $J$ has discrete spectrum (as explained in
Section~2, when $J_0$ has deficiency indices $(1,1)$, this is always
the case), this work treats the inverse spectral problem of
reconstructing, from the spectra of $J$ and $\widetilde{J}$, the
matrix (\ref{eq:jm-0}) and the ``boundary condition at infinity''
defining the self-adjoint extension $J$ if necessary (i.\,e. if $J_0$
is not essentially self-adjoint,
cf. \cite[Sec.\,2]{delrio-kudryavtsev-II}). To solve this inverse
problem, one should elucidate the distribution of the perturbed
spectrum relative to the unperturbed one and determine the necessary
input data for recovering the matrix.  An important point to note is
that this work provides necessary and sufficient conditions for two
sequences to be the spectra of $J$ and $\widetilde{J}$. Also, we
discuss (the lack of) uniqueness of the reconstruction.

Although the two spectra inverse problem for the rank-one perturbation
family of Jacobi operators has been thoroughly studied (see for
instance \cite{MR49:9676,MR0221315,weder-silva,MR1643529} and
\cite{Chu-Golub,deBoor-Golub,MR1616422,MR2102477} for the case of
finite matrices), there is scarce literature dealing with inverse
problems of other kind of perturbations
(cf. \cite{delrio-kudryavtsev-II}).

The motivation for this work is the inverse spectral problem studied in
\cite{Ram} and \cite{delrio-kudryavtsev} which is in its turn related
with the physical problem of measuring micro-masses with the help of
micro-cantilevers
\cite{spletzer-et-al1,spletzer-et-al2}. Micro-cantilevers are modeled
by spring-mass systems whose masses and spring constants are
determined by the mechanical parameters of the micro-cantilevers.

In this work we consider the semi-infinite mass-spring system given in
Fig. 1. with masses $\{m_j\}_{j=1}^\infty$ and spring constants
$\{k_j\}_{j=1}^\infty$.  This system is modeled by the Jacobi matrix
(\ref{eq:jm-0}) with
\begin{equation*}
q_j = -\frac{k_{j+1}+k_j}{m_j}\,, \qquad
b_j=\frac{k_{j+1}}{\sqrt{m_j m_{j+1}}}\,,
\qquad j\in\mathbb{N}\,.
\end{equation*}
In \cite{MR2102477,mono-marchenko} it is explained how to deduce these
formulae. Since $J$ is considered to have discrete spectrum, the
movement of the system is a superposition of harmonic oscillations
whose frequencies are the square roots of the modules of the
eigenvalues.
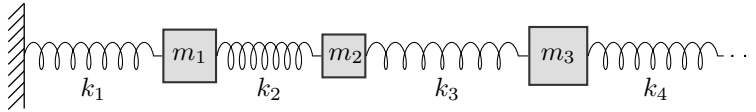
\begin{figure}[h]
\begin{center}
\begin{tikzpicture}
  [mass1/.style={rectangle,draw=black!80,fill=black!13,thick,inner sep=0pt,
   minimum size=7mm},
   mass2/.style={rectangle,draw=black!80,fill=black!13,thick,inner sep=0pt,
   minimum size=5.7mm},
   mass3/.style={rectangle,draw=black!80,fill=black!13,thick,inner sep=0pt,
   minimum size=7.7mm},
   wall/.style={postaction={draw,decorate,decoration={border,angle=-45,
   amplitude=0.3cm,segment length=1.5mm}}}]
  \node (mass3) at (7.1,1) [mass3] {\footnotesize$m_3$};
  \node (mass2) at (4.25,1) [mass2] {\footnotesize$\,m_2$};
  \node (mass1) at (2.2,1) [mass1] {\footnotesize$m_1$};
\draw[decorate,decoration={coil,aspect=0.4,segment
  length=2.1mm,amplitude=1.8mm}] (0,1) -- node[below=4pt]
{\footnotesize$k_1$} (mass1);
\draw[decorate,decoration={coil,aspect=0.4,segment
  length=1.5mm,amplitude=1.8mm}] (mass1) -- node[below=4pt]
{\footnotesize$k_2$} (mass2);
\draw[decorate,decoration={coil,aspect=0.4,segment
  length=2.5mm,amplitude=1.8mm}] (mass2) -- node[below=4pt]
{\footnotesize$k_3$} (mass3);
\draw[decorate,decoration={coil,aspect=0.4,segment
  length=2.1mm,amplitude=1.8mm}] (mass3) -- node[below=4pt]
{\footnotesize$k_4$} (9.3,1);
\draw[line width=.8pt,loosely dotted] (9.4,1) -- (9.8,1);
\draw[line width=.5pt,wall](0,1.7)--(0,0.3);
\end{tikzpicture}
\end{center}
\caption{Semi-infinite mass-spring system}\label{fig:1}
\end{figure}
The modified mass-spring system corresponding to the perturbed
operator $\widetilde{J}$ is obtained by changing the first mass by
$\Delta m=m_1(\theta^{-2}-1)$ and the first spring by $\Delta
k=-hm_1$ (see Fig. 2). Here we also consider negative values of
$\Delta m$ and $\Delta k$ which correspond to $\theta>1$ and $h<0$,
respectively.
\begin{figure}[h]
\begin{center}
\begin{tikzpicture}
  [mass1/.style={rectangle,draw=black!80,fill=black!13,thick,inner sep=0pt,
   minimum size=7mm},
   mass2/.style={rectangle,draw=black!80,fill=black!13,thick,inner sep=0pt,
   minimum size=5.7mm},
   mass3/.style={rectangle,draw=black!80,fill=black!13,thick,inner sep=0pt,
   minimum size=7.7mm},
   dmass/.style={rectangle,draw=black!80,fill=black!13,thick,inner sep=0pt,
   minimum size=5mm},
   wall/.style={postaction={draw,decorate,decoration={border,angle=-45,
   amplitude=0.3cm,segment length=1.5mm}}}]
  \node (mass3) at (7.1,1) [mass3] {\footnotesize$m_3$};
  \node (mass2) at (4.25,1) [mass2] {\footnotesize$\,m_2$};
  \node (mass1) at (2.2,1) [mass1] {\footnotesize$m_1$};
  \node (dmass) at (2.2,1.6) [dmass] {\scriptsize$\,\Delta m\,$};
\draw[decorate,decoration={coil,aspect=0.4,segment
  length=1.9mm,amplitude=1.8mm}] (0,1.6) -- node[above=4pt]
{\footnotesize$\Delta k$} (dmass);
\draw[decorate,decoration={coil,aspect=0.4,segment
  length=2.1mm,amplitude=1.8mm}] (0,1) -- node[below=4pt]
{\footnotesize$k_1$}  (mass1);
\draw[decorate,decoration={coil,aspect=0.4,segment
  length=1.5mm,amplitude=1.8mm}] (mass1) -- node[below=4pt]
{\footnotesize$k_2$} (mass2);
\draw[decorate,decoration={coil,aspect=0.4,segment
  length=2.5mm,amplitude=1.8mm}] (mass2) -- node[below=4pt]
{\footnotesize$k_3$} (mass3);
\draw[decorate,decoration={coil,aspect=0.4,segment
  length=2.1mm,amplitude=1.8mm}] (mass3) -- node[below=4pt]
{\footnotesize$k_4$} (9.3,1);
\draw[line width=.8pt,loosely dotted] (9.4,1) -- (9.8,1);
\draw[line width=.5pt,wall](0,2.1)--(0,0.7);
\end{tikzpicture}
\end{center}
\caption{Perturbed semi-infinite mass-spring system}\label{fig:2}
\end{figure}
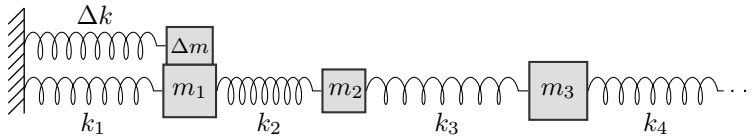
Note that the perturbation involved here is the result of the combined
effect of a rank-one perturbation (studied thoroughly in
\cite{weder-silva}) and the particular rank-two perturbation studied
in \cite{delrio-kudryavtsev-II}. However, most of the results obtained
here cannot be found from the results in \cite{weder-silva} and
\cite{delrio-kudryavtsev-II}, and require their own proof. Moreover,
it turns out that one can single-out classes of isospectral operators
within the two parameter perturbation family considered in this work
that were not studied before.

The paper is organized as follows. In Section 2 we fix the notation,
lay down a convention for enumerating sequences and recall some
results of the inverse spectral theory for Jacobi operators. Section 3
gives a detailed spectral analysis of the family of perturbed Jacobi
operators. The solution of the two spectra inverse problem for $J$ and
$\widetilde{J}$ is given in Section 4. This section also discusses the
non-uniqueness of the reconstruction and gives some characterization of
isospectral operators in the perturbation family under consideration.

\section{A review on inverse spectral theory for Jacobi operators}
\label{sec:preliminaries}
Let us denote by $\sigma(J)$ the spectrum of $J$ and consider the
spectral resolution of the identity $E$ for $J$ given by the spectral
theorem. Then the spectral function $\rho$ of $J$ is defined by
\begin{equation}
  \label{eq:rho-def}
  \rho(t):=\inner{\delta_1}{E(t)\delta_1}\,.
\end{equation}
All the moments of $\rho$ exist \cite[Thm.\,4.1.3]{MR0184042}, that
is, for all $k\in\mathbb{N}\cup\{0\}$,
\begin{equation*}
  s_k=\int_{\mathbb{R}} t^kd\rho(t)\in\mathbb{R}\,.
\end{equation*}
Moreover, since $J$ turns out to be simple with $\delta_1$ being a
cyclic vector, the operator of multiplication by the independent
variable in $L_2(\mathbb{R},\rho)$ (defined on the maximal domain) is
unitarily equivalent to $J$.

Alongside the spectral function we consider the corresponding
Weyl $m$-function given by
\begin{equation}
  \label{eq:weyl-function}
  m(\zeta):=\inner{\delta_1}{(J-\zeta I)^{-1}\delta_1}=\int_{\mathbb{R}}
\frac{d\rho(t)}{t-\zeta}\,,\qquad \zeta\not\in\sigma(J)\,.
\end{equation}

Because of the inverse Stieltjes transform one uniquely recovers
$\rho$ from $m$, so $\rho$ and $m$ are in one-to-one correspondence.

%%%%%%%%%%%%%%%%% do we use this?
The Weyl $m$-function has the following asymptotic behavior
\begin{equation}
  \label{eq:m-asympt}
  m(\zeta)=-\frac{1}{\zeta}-\frac{q_1}{\zeta^2}
  -\frac{b_1^2+q_1^2}{\zeta^3}
  +O(\zeta^{-4})\,,
\end{equation}
as $\zeta\to\infty$ with $\im \zeta\ge \epsilon$, $\epsilon>0$
(see \cite[Eq.\,1.5]{MR1616422} and \cite[Eq.\,2.10]{weder-silva}).
%%%%%%%%%%%%%%%%

The inverse spectral theory for the Jacobi operator $J$ is centered on
the fact that the Weyl $m$-function (or, equivalently, $\rho$) uniquely
determines the matrix (\ref{eq:jm-0}) and the boundary condition at
infinity that defines the self-adjoint extension if necessary. Indeed,
for recovering the matrix one may use a method based on a discrete
Riccati equation (see \cite[Eq.\,2.15]{MR1616422},
\cite[Eq.\,2.23]{MR1643529}) or the method of orthonormalization of
the polynomial sequence $\{t^k\}_{k=0}^\infty$ in $L_2(\mathbb{R},\rho)$
\cite[Chap.\,7,\,Sec.\,1.5]{MR0222718}. If (\ref{eq:jm-0}) is the
matrix representation of a non-self-adjoint operator, then the
condition at infinity may be found by the method exposed in
\cite[Sec.\,2]{weder-silva}.

In this work we restrict our considerations to the case of $\sigma(J)$
being discrete, viz., $\sigma_{\rm ess}(J)=\emptyset$. It is well
known that this is always the case when $J_0$ is not essentially
self-adjoint \cite[Thm.\,4.11]{MR1627806},
\cite[Lem.\,2.19]{MR1711536}. The discreteness of $\sigma(J)$ implies
that (\ref{eq:rho-def}) can be written as follows
\begin{equation}
  \label{eq:rho-discrete}
   \rho(t)=\sum_{\lambda_k< t}\frac{1}{\alpha_k}\,,
\end{equation}
where the coefficients $\{\alpha_k\}_k$ are called the normalizing
constants. From (\ref{eq:weyl-function}) and (\ref{eq:rho-discrete})
it follows that
\begin{equation}
  \label{eq:m-discrete}
  m(\zeta)=\sum_{k}\frac{1}{\alpha_k(\lambda_k-\zeta)}\,.
\end{equation}
The function $m$ is meromorphic, and, since it is also Herglotz, its
zeros and poles interlace, i.\,e., between two contiguous zeros there
is only one pole and between two contiguous poles there is only one
zero (see the proof of \cite[Chap.\,7,\,Thm.\,1]{MR589888}).

Now, in the subspace $\delta_1^\perp$ of $l_2(\mathbb{N})$, consider the
operator $J_{\rm T}$ which is the restriction of $J$ to
$\dom(J)\cap\delta_1^\perp$.  Note that $J_T$ is a self-adjoint
extension of the operator whose matrix representation with respect to
the basis $\{\delta_k\}_{k=2}^\infty$ of
the space $\delta_1^\perp$ is (\ref{eq:jm-0}) with the first
column and row removed. The following proposition is well known (see
for instance \cite{weder-silva}).
\begin{proposition}
\label{prop:interlace-truncated}
Under the assumption that $\sigma(J)$ is discrete, $\sigma(J)$ and
$\sigma(J_{\rm T})$ interlace. Moreover  $\sigma(J)$ coincides with
the set of poles of the function $m$ and $\sigma(J_{\rm T})$ is the
set of its zeros.
\end{proposition}
\begin{proof}
  Clearly, one should only establish that the zeros and poles of $m$
  are as stated in the proposition. But this is a straightforward
  conclusion from the definition of the Weyl $m$-function and the
  formula
\begin{equation}
  \label{eq:first-riccati}
  b_1^2m_{\rm T}(\zeta)=q_1-\zeta-\frac{1}{m(\zeta)}\,,
\end{equation}
where $m_{\rm T}$ is the Weyl $m$-function corresponding to $J_{\rm
  T}$. Equation (\ref{eq:first-riccati}) is a particular case of
\cite[Eq.\,2.15]{MR1616422} or \cite[Eq.\,2.23]{MR1643529}.
\end{proof}

\noindent\textbf{(C1) Convention for enumerating a sequence.} Let $S$ be
an infinite countable set of real numbers without finite points of
accumulation and $M$ an infinite subset of consecutive integers such
that there is a strictly increasing function $f:M\to S$ such that
$f^{-1}(0)=0$. We write $S=\{\lambda_k\}_{k\in M}$, where
$\lambda_k=f(k)$. Note that $M$ is semi-bounded from above (below) if
and only if the same holds for $S$ and that in $\{\lambda_k\}_{k\in
  M}$ only $\lambda_0$ is allowed to be zero.
\begin{remark}
\label{rem:true-interlacing}
Clearly, if two real sequences $S$, $S'$ without finite accumulation
points interlace, then one always can find $M$ and functions $f:M\to
S$ and $f':M\to S'$ with the properties given in our convention (C1) such
that, for any $k\in M$, either
\begin{equation*}
  \lambda_k<\lambda_k'<\lambda_{k+1}\quad\text{ or }\quad
\lambda_k'<\lambda_k<\lambda_{k+1}'\,,
\end{equation*}
where $\lambda_k=f(k)$ and $\lambda_k'=f'(k)$. If $S$ is not
semi-bounded, then both possibilities hold simultaneously.
\end{remark}
The proof of the following proposition can be found in
\cite[Lem.\,4.1]{delrio-kudryavtsev-II} and
\cite[Sec.\,4]{weder-silva} and the starting point for it is
\cite[Chap.\,7,\,Thm.\,1]{MR589888}.
\begin{proposition}
  \label{prop:m-weyl-krein-representation}
  Let $J$ have discrete spectrum and assume that
  $\sigma(J)=\{\lambda_k\}_{k\in M}$, and
  $\sigma(J_{\rm T})=\{\eta_k\}_{k\in M}$. Then
  \begin{equation}
 \label{eq:levin-herglotz-gen}
    m(\zeta)=C \frac{\zeta-\eta_0}{\zeta-\lambda_0}
  \prod_{\substack{k\in M\\k\ne 0}} \left(1-\frac{\zeta}{\eta_k}\right)
  \left(1-\frac{\zeta}{\lambda_k}\right)^{-1}\,,
  \end{equation}
Moreover, $C<0$ and
\begin{equation}
  \label{eq:enum-zeros-poles-alt}
  \eta_k<\lambda_k<\eta_{k+1}\,,\quad\forall k\in M\,,
\end{equation}
if $\sigma(J)$ is semi-bounded from above,
while, $C>0$ and
\begin{equation}
  \label{eq:enum-zeros-poles}
  \lambda_k<\eta_k<\lambda_{k+1}\,,\quad\forall k\in M\,,
\end{equation}
otherwise.
\end{proposition}
\section{Direct spectral analysis for $J$ and $\widetilde{J}$}
\label{sec:direct-spectral-analysis-general-case}
Let $J$ and $\widetilde{J}$ be the operators defined in the
Introduction.  Since $J_{\rm T}=\widetilde{J}_{\rm T}$, where
$\widetilde{J}_{\rm T}$ is the operator in the space $\delta_1^\perp$
obtained by restricting $\widetilde{J}$ to
$\dom(\widetilde{J})\cap\delta_1^\perp$, one obtains from
(\ref{eq:first-riccati}) that
\begin{equation}
  \label{eq:aux-m-m-theta2}
    \theta^2\left(\zeta+\frac{1}{m(\zeta)}+h\right)=
\zeta+\frac{1}{\widetilde{m}(\zeta)}\,,
\end{equation}
where $\widetilde{m}$ is the Weyl $m$-function corresponding to
$\widetilde{J}$.  Let us define the function
\begin{equation}
  \label{eq:m-goth-def2}
  \mathfrak{m}(\zeta):=\frac{m(\zeta)}{\widetilde{m}(\zeta)}\,
\end{equation}
Immediately from (\ref{eq:aux-m-m-theta2}) one proves the following
proposition. Prior to stating it, in order to simplify the writing of
some expressions, let us introduce a constant that will be used
recurrently throughout the paper.
\begin{equation}
  \label{eq:gamma-def}
  \gamma:=\frac{\theta^2h}{1-\theta^2}\,.
\end{equation}
\begin{proposition}
  \label{prop:zeros-poles2}
 Consider the Jacobi operator $J$ and the operator $\widetilde{J}$ as
 given in (\ref{eq:def-tilde-j})
 with $\theta\ne 1$. If
  $J$ has discrete spectrum, then
  \begin{enumerate}[i)]
  \item the set of poles of
  $\mathfrak{m}$ is a subset of $\sigma(J)$ and the set of zeros is
  contained in  $\sigma(\widetilde{J})$,
\item $\displaystyle\gamma\in\sigma(J)$ if and only if
  $\displaystyle\gamma\in\sigma(\widetilde{J})$,
\item the sets $\sigma(J)$ and $\sigma(\widetilde{J})$ can
  intersect only at $\displaystyle\gamma$.
  \end{enumerate}
\end{proposition}
The following alternative expression for $\mathfrak{m}$:
\begin{equation}
  \label{eq:m-through-m2}
  \mathfrak{m}(\zeta)=
(\theta^2-1)\left(\zeta-\gamma\right)
m(\zeta)+\theta^2\,,
\end{equation}
which is obtained by combining (\ref{eq:aux-m-m-theta2}) and
(\ref{eq:m-goth-def2}), is the main ingredient in the proof of the
following proposition.
\begin{proposition}
  \label{prop:interlacing2}
Consider the Jacobi operator $J$ and the operator $\widetilde{J}$ as
 given in (\ref{eq:def-tilde-j})
 with $\theta\ne 1$. If  $J$ has discrete
  spectrum, then the spectra
  $\sigma(J)$, $\sigma(\widetilde{J})$ interlace in the intervals
  $(\gamma,+\infty)$ and
  $(-\infty,\gamma)$. Moreover,
  $\sigma(\widetilde{J})$ in the interval
  $(\gamma,+\infty)$, respectively
  $(-\infty,\gamma)$, is shifted with respect to
  $\sigma(J)$ to the left, respectively right, when $\theta<1$ and to
  the right, respectively left, when $\theta>1$.
\end{proposition}
\begin{remark}
  \label{rem:semi-intervals-cases}
  The set $\sigma(J)\cap(\gamma,+\infty)$,
  respectively $\sigma(J)\cap(-\infty,\gamma)$,
  may be empty and, then, there is no spectrum of $\widetilde{J}$ in
  $(\gamma,+\infty)$, respectively
  $(-\infty,\gamma)$. If $\lambda$ is the only
  element in $\sigma(J)\cap(\gamma,+\infty)$,
  respectively $\sigma(J)\cap(-\infty,\gamma)$,
  then there is exactly one element of $\sigma(\widetilde{J})$ in
  $(\gamma,+\infty)$, respectively
  $(-\infty,\gamma)$.
\end{remark}
\begin{proof}
  Let us first prove that between two contiguous eigenvalues of $J$
  there is exactly one eigenvalue of $\widetilde{J}$.  Assume that
  $\theta>1$ and consider two contiguous eigenvalues
  $\lambda,\widehat{\lambda}$ of $J$ such that
  $\gamma<\lambda<\widehat{\lambda}$. Then, by
  (\ref{eq:m-discrete}) and (\ref{eq:m-through-m2}), one
  has
  \begin{equation*}
      \lim_{\substack{t\to\widehat{\lambda}^- \\
         t\in\mathbb{R}}}\mathfrak{m}(t)
    =+\infty\qquad
    \lim_{\substack{t\to\lambda^+ \\ t\in\mathbb{R}}}\mathfrak{m}(t)=-\infty\,.
  \end{equation*}
  The function $\mathfrak{m}\upharpoonright_{\mathbb{R}}$, should
  cross the 0-axis in $(\lambda,\widehat{\lambda})$ an odd number of
  times. Actually, it crosses the 0-axis only once. Indeed, if one
  assumes that $\mathfrak{m}\upharpoonright_{\mathbb{R}}$ crosses the
  0-axis three or more times as in Fig. 3 (a), then, in view of
  Propositions \ref{prop:interlace-truncated} and
  \ref{prop:zeros-poles2}, there would be at least two elements of
  $\sigma(J_{\rm T})$ in $(\lambda,\widehat{\lambda})$. Note that one
  crossing of the 0-axis and a tangential touch of it as in Fig. 3 (b)
  and (c) is also impossible since the poles of $\widetilde{m}$ are
  simple. Analogously, between two contiguous eigenvalues of
  $\widetilde{J}$, $1/\mathfrak{m}\upharpoonright_{\mathbb{R}}$ crosses
  the 0-axis exactly once. Thus, by means of Proposition
  \ref{prop:zeros-poles2}, the interlacing of $\sigma(J)$ and
  $\sigma(\widetilde{J})$ in $(\gamma,+\infty)$ has been proven.
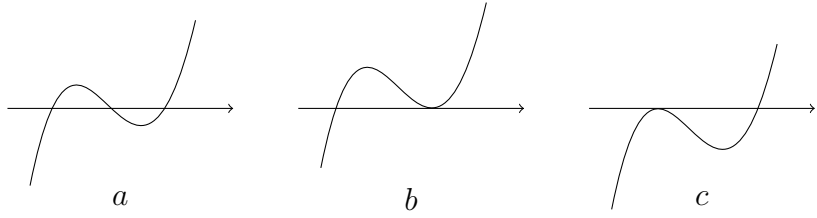
\begin{figure}[h]
\vspace*{-30pt}
\begin{center}
  \begin{tikzpicture}
    \draw[yshift=-15] (-.7,-.5) .. controls (0,3) and (.65,-1.9)
    .. (1.5,1.7);
    \draw[->] (-1,0) -- (2,0);
    \path (0.5,-1.2) node {$a$};
    \draw[xshift=110,yshift=-8.3] (-.7,-.5) .. controls (0,3) and (.65,-1.9)
    .. (1.5,1.7);
    \draw[->,xshift=110] (-1,0) -- (2,0);
    \path[xshift=110] (0.5,-1.2) node {$b$};
    \draw[xshift=220,yshift=-24] (-.7,-.5) .. controls (0,3) and (.65,-1.9)
    .. (1.5,1.7);
    \draw[->,xshift=220] (-1,0) -- (2,0);
    \path[xshift=220] (0.5,-1.2) node {$c$};
  \end{tikzpicture}
\end{center}
\vspace*{-30pt}
\caption{Impossible crossings of the 0-axis by $\mathfrak{m}$}\label{fig:3}
\end{figure}

When $\theta<1$, one has
  \begin{equation*}
      \lim_{\substack{t\to\widehat{\lambda}^- \\
         t\in\mathbb{R}}}\mathfrak{m}(t)
    =-\infty\qquad
    \lim_{\substack{t\to\lambda^+ \\ t\in\mathbb{R}}}\mathfrak{m}(t)=+\infty\,.
  \end{equation*}
  and by the same reasoning used above the interlacing of the spectra
  in $(\gamma,+\infty)$ is established. The
  interlacing in $(-\infty,\gamma)$ is proven
  analogously.

  Let us now prove the second assertion of the proposition. To this
  end suppose first that
  $\gamma\not\in\sigma(J)$ and observe that,
  under this assumption, (\ref{eq:m-through-m2}) implies that
  \begin{equation}
    \label{eq:n-goth-at-0}
    \mathfrak{m}\left(\gamma\right)=\theta^2\,.
  \end{equation}
  Let us now assume that the contiguous eigenvalues
  $\lambda,\widehat{\lambda}$ of $J$ are such that
  \begin{equation*}
   \lambda<\gamma <\widehat{\lambda}\,.
  \end{equation*}
  Under the premise that $\theta>1$, we have
  \begin{equation}
    \label{eq:n-goth-asympt-+}
   \lim_{\substack{t\to\widehat{\lambda}^- \\
         t\in\mathbb{R}}}\mathfrak{m}(t)
    =+\infty\qquad
    \lim_{\substack{t\to\lambda^+ \\ t\in\mathbb{R}}}\mathfrak{m}(t)=+\infty\,.
  \end{equation}
  In view of (\ref{eq:n-goth-at-0}) and (\ref{eq:n-goth-asympt-+}), if
  $\mathfrak{m}\upharpoonright_{\mathbb{R}}$ crosses the 0-axis one time in
  the interval $(\lambda,\gamma)$, it should
  cross it in $(\lambda,\gamma)$ at least
  twice. The same is true for the interval
  $(\gamma,\widehat{\lambda})$. Note that
  $\mathfrak{m}\upharpoonright_{\mathbb{R}}$ cannot tangentially touch the
  0-axis due to the simplicity of its zeros. So, the assumption that
  $\mathfrak{m}\upharpoonright_{\mathbb{R}}$ crosses the 0-axis, from what
  has already been proven above, would imply that in
  $(\lambda,\gamma)$, respectively
  $(\gamma,\widehat{\lambda})$, there is at
  least one eigenvalue of $J$, which contradicts the fact that
  $\lambda$ and $\widehat{\lambda}$ are contiguous. Thus, there is no
  crossing of the 0-axis by $\mathfrak{m}\upharpoonright_{\mathbb{R}}$ in
  the interval $(\lambda,\widehat{\lambda})$, which means the
  absence of eigenvalues of $\widetilde{J}$ in
  $(\lambda,\widehat{\lambda})$. If now $\theta<1$, instead of
  (\ref{eq:n-goth-asympt-+}), one has
  \begin{equation*}
   \lim_{\substack{t\to\widehat{\lambda}^- \\
         t\in\mathbb{R}}}\mathfrak{m}(t)
    =-\infty\qquad
    \lim_{\substack{t\to\lambda^+ \\ t\in\mathbb{R}}}\mathfrak{m}(t)=-\infty\,.
  \end{equation*}
 From this asymptotic behavior, together with (\ref{eq:n-goth-at-0}) and a
 similar reasoning as the one given above, it follows that
 $\mathfrak{m}\upharpoonright_{\mathbb{R}}$ crosses the 0-axis exactly once
 in $(\lambda,\gamma)$ and once in
 $(\gamma,\widehat{\lambda})$.

 The case when $\gamma$ is in $\sigma(J)$ is
 treated analogously. Here one only has to take into account two
 things: firstly that now
 \begin{equation}
   \label{eq:n-goth-at-0-special-case}
  \mathfrak{m}\left(\gamma\right)=
\theta^2+(\theta^2-1)\res_{\zeta=\gamma}m(\zeta)
 \end{equation}
 and secondly, that, since
 $-\left[\res_{\zeta=\gamma}m(\zeta)\right]^{-1}$
 is the normalizing constant of $J$ corresponding to the eigenvalue
 $\gamma$ (see (\ref{eq:m-discrete})), one has
\begin{equation*}
  \mathfrak{m}\left(\gamma\right)>0
\end{equation*}
either when $\theta>1$ or $\theta<1$.
\end{proof}
\begin{remark}
  Although, the case $\theta=1$ reduces to an additive rank-one
  perturbation, the well known interlacing property (see for instance
  the proof of \cite[Thm.\,3.3]{weder-silva}) cannot be obtained from
  Proposition \ref{prop:interlacing2} by a limiting procedure since
  the limit of $\gamma(\theta)$ when $\theta\to 1$ does not exist (see
  (\ref{eq:gamma-def})).
\end{remark}
\begin{remark}
  \label{rem:special-interlace}
  Let the positive number $\theta\ne 1$ and $h\in\mathbb{R}$. It is
  straightforward to verify that, for $\sigma(J)$ and
  $\sigma(\widetilde{J})$, there exist a set $M$ and functions
  $f:M\to\sigma(J)$ and $\widetilde{f}:M\to\sigma(\widetilde{J})$, with
  the properties given in our convention (C1) for enumerating sequences,
  such that the following conditions hold under the assumption that
  $\lambda_k=f(k)$ and $\mu_k=\widetilde{f}(k)$:
  \begin{equation}
    \label{eq:interlacing-theta-big}
    \lambda_k<\mu_k<\lambda_{k+1} \quad\text{in}
\ \left(\gamma,+\infty\right)\,,
    \qquad
    \lambda_{k-1}<\mu_k<\lambda_k \quad\text{in}
\ \left(-\infty,\gamma\right)\,,
  \end{equation}
when $\theta>1$, and
\begin{equation}
  \label{eq:interlacing-theta-small}
  \mu_k<\lambda_k<\mu_{k+1}\quad\text{in}
\ \left(\gamma,+\infty\right)\,,
  \qquad
  \mu_{k-1}<\lambda_k<\mu_k \quad\text{in}
\ \left(-\infty,\gamma\right)\,,
\end{equation}
if $\theta<1$. Here, implicitly, the intersection of $\sigma(J)$ with
the semi-infinite intervals is not empty, but we are also considering
the case when the intersection with one of the semi-infinite
intervals is empty (see Remark~\ref{rem:semi-intervals-cases}). Also,
we are not excluding the case when $\gamma$ is
in $\sigma(J)$ for which there is $k_0\in M$ such that
$\lambda_{k_0}=\mu_{k_0}=\gamma$.
\end{remark}
\begin{proposition}
  \label{prop:convergence-eigenvalues2}
  Suppose that $h\in\mathbb{R}$ is such that if $\theta=1$ then $h\ne
  0$.  Let $J$ have discrete spectrum and assume that
  $\sigma(J)=\{\lambda_k\}_{k\in M}$ and
  $\sigma(\widetilde{J})=\{\mu_k\}_{k\in M}$, where the sequences have
  been arranged according to Remark \ref{rem:special-interlace} if
  $\theta\ne 1$ and according to Remark~\ref{rem:true-interlacing}
  otherwise. Then
 \begin{equation}
   \label{eq:convergence-eigenvalues}
    \sum_{k\in M}(\mu_k-\lambda_k)=h+q_1(\theta^2-1)
  \end{equation}
\end{proposition}
\begin{proof}
  Consider  the sequence
  $\{\eta_k\}_{k\in M}$ being the spectrum of $\widehat{J}$, where
  \begin{equation*}
    \widehat{J}:=J+h\inner{\delta_1}{\cdot}\delta_1\,.
  \end{equation*}
In the proof of \cite[Thm.\,3.4]{weder-silva} it is shown that
\begin{equation*}
  \sum_{k\in M}(\eta_k-\lambda_k)=h\,,
\end{equation*}
where $\eta_k>\lambda_k$ for all $k\in M$ when $h>0$ and
$\eta_k\le\lambda_k$ for all $k\in M$ otherwise. On the other hand, by
\cite[Prop.\,4.1]{delrio-kudryavtsev-II}, one has
\begin{equation*}
   \sum_{k\in M}(\mu_k-\eta_k)=q_1(\theta^2-1)\,,
\end{equation*}
where the enumeration obeys \cite[Remark
5]{delrio-kudryavtsev-II} if $\theta\ne 1$.

Consider a sequence $\{M_n\}_{n=1}^\infty$
of subsets of $M$, such that $M_n\subset M_{n+1}$ and $\cup_nM_n=M$.
Then the assertion follows from the
linearity of the limit
\begin{equation*}
  \lim_{n\to\infty}\left[\sum_{k\in M_n}(\mu_k-\eta_k)+
\sum_{k\in M_n}(\eta_k-\lambda_k)\right]\,,
\end{equation*}
as soon as one notices that the enumeration has
been done
according to Remark~\ref{rem:special-interlace} when $\theta\ne 1$.
\end{proof}
\begin{proposition}
  \label{prop:m-goth-actual-form2}
Suppose that $h\in\mathbb{R}$ is such that if $\theta=1$ then $h\ne
  0$.  Let the Jacobi operator $J$ have discrete spectrum and assume that
  $\sigma(J)=\{\lambda_k\}_{k\in M}$ and
  $\sigma(\widetilde{J})=\{\mu_k\}_{k\in M}$, where $\widetilde{J}$ is
  given by (\ref{eq:def-tilde-j}), and the sequences have
  been arranged according to Remark \ref{rem:special-interlace} if
  $\theta\ne 1$ and according to Remark~\ref{rem:true-interlacing}
  otherwise. Then,
  \begin{equation*}
     \mathfrak{m}(\zeta)=
      \prod\limits_{k\in M}\frac{\zeta-\mu_k}{\zeta-\lambda_k}\,.
  \end{equation*}
\end{proposition}
\begin{proof}
  When $\theta=1$ the assertion follows from the proof of
  \cite[Thm.\,3.4]{weder-silva} If $\theta\ne 1$, the proof repeats
  the one of \cite[Prop.\,4.2]{delrio-kudryavtsev-II}, so we omit some
  details that the reader can reestablish from
  \cite[Prop.\,4.2]{delrio-kudryavtsev-II} if necessary.

From Proposition~\ref{prop:m-weyl-krein-representation} and
(\ref{eq:m-goth-def2}) it follows that
\begin{equation*}
   \mathfrak{m}(\zeta)=C\frac{\zeta-\mu_0}{\zeta-\lambda_0}
   \prod_{\substack{k\in M\\k\ne 0}}
\left(1-\frac{\zeta}{\mu_k}\right)
  \left(1-\frac{\zeta}{\lambda_k}\right)^{-1}\,.
\end{equation*}
By Proposition \ref{prop:convergence-eigenvalues2}, one actually has
\begin{equation}
  \label{eq:m-goth-two-products}
  \mathfrak{m}(\zeta)=C\frac{\zeta-\mu_0}{\zeta-\lambda_0}
\prod_{\substack{k\in M\\k\ne 0}}\frac{\lambda_k}{\mu_k}
\prod_{\substack{k\in M\\k\ne 0}}\frac{\zeta-\mu_k}{\zeta-\lambda_k}
\end{equation}
Indeed, (\ref{eq:convergence-eigenvalues}) implies the convergence of
the products in (\ref{eq:m-goth-two-products}).

Now, the assertion of the proposition follows from
\begin{equation}
  \label{eq:limits-of-m-goth}
  \lim_{\substack{\zeta\to\infty \\ \im \zeta\ge\epsilon>0}}
    \mathfrak{m}(\zeta)=1\quad\text{ and }\quad
    \lim_{\substack{\zeta\to\infty \\
        \im \zeta\ge\epsilon}}
    \prod_{k\in M}
    \frac{\zeta-\mu_k}{\zeta-\lambda_k}=1\,.
\end{equation}
The first limit is obtained from (\ref{eq:m-asympt}) and
(\ref{eq:m-through-m2}). The second one is a consequence of the
uniform convergence of
\begin{equation*}
  \prod_{k\in M}
    \frac{\zeta-\mu_k}{\zeta-\lambda_k}
\end{equation*}
in compacts of $\mathbb{C}\setminus\mathbb{R}$, which, in its turn, can be
proven on the basis of (\ref{eq:convergence-eigenvalues}).
\end{proof}

\section{Inverse spectral analysis for $J$ and $\widetilde{J}$}
\label{sec:reconstruction}
In this section we give results on reconstruction of the operator $J$
from its spectrum and the one of $\widetilde{J}$. Additionally, we
provide necessary and sufficient conditions for two sequences to be
the spectra of the operators $J$ and $\widetilde{J}$. Finally, we
discuss isospectral operators within the perturbed family of Jacobi
operators.
\begin{theorem}
    \label{prop:reconstruction2}
    Let the Jacobi operator $J$ have discrete spectrum and
    $\widetilde{J}$ be as in (\ref{eq:def-tilde-j}) with $\theta\ne
    1$. If $\gamma$ is not in $\sigma(J)$, then the sets $\sigma(J)$,
    $\sigma(\widetilde{J})$, and the constant $\gamma$ uniquely
    determine the matrix (\ref{eq:jm-0}), the parameters $\theta$ and
    $h$, and the boundary condition at infinity if necessary
    (i.\,e. if $J_0$ turns out to be non-essentially self-adjoint).
\end{theorem}
\begin{proof}
 In view of what has been said in Section~\ref{sec:preliminaries}, it
 suffices to show that the input data uniquely determine the Weyl
 $m$-function of $J$, and the parameters $\theta$ and $h$.

 On the basis of Proposition~\ref{prop:m-goth-actual-form2}, one
 construct $\mathfrak{m}$ from the sets $\sigma(J)$ and
 $\sigma(\widetilde{J})$. Then, since $\gamma\not\in\sigma(J)$, it follows
 from (\ref{eq:m-through-m2}) that $\mathfrak{m}(\gamma)=\theta^2$. Now,
 the constants $\gamma$ and $\theta$ allow to find $h$. Finally, by means
 of (\ref{eq:m-through-m2}), one determines the function $m$.
\end{proof}
\begin{theorem}
    \label{prop:reconstruction3}
    Let the Jacobi operator $J$ have discrete spectrum and
    $\widetilde{J}$ be as in (\ref{eq:def-tilde-j}) with $\theta\ne
    1$. Assuming that $\gamma$ is in $\sigma(J)$, suppose that one is
    given the sets $\sigma(J)$, $\sigma(\widetilde{J})$ and one of the
    following constants
    \begin{equation*}
      \text{(a)}\ \theta,\qquad\text{(b) the normalizing constant
        corresponding to}\ \gamma,\qquad\text{(c)}\ h,
    \end{equation*}
    then one recovers uniquely the matrix (\ref{eq:jm-0}), the
    constant $h$ in case (a), $\theta$ and $h$ in case (b), $\theta$
    in case (c), and the boundary condition at infinity if necessary
    (i.\,e. if $J_0$ turns out to be non-essentially self-adjoint).
\end{theorem}
\begin{proof}
  The proof is similar to the one of
  Theorem~\ref{prop:reconstruction2}. The sets $\sigma(J)$ and
  $\sigma(\widetilde{J})$ determine $\mathfrak{m}$ and, then, one
  should obtain from it the function $m$ using either the constant
  $\theta$ or the normalizing constant corresponding to $\gamma$. From
  Proposition~\ref{prop:zeros-poles2} it follows that
  \begin{equation*}
    \sigma(J)\cap\sigma(\widetilde{J})=\{\gamma\}\,.
  \end{equation*}
  Thus $\theta$ or $h$ determine $\theta$ and $h$. On the other hand,
  from (\ref{eq:m-through-m2}) and taking into account that
  $\gamma\in\sigma(J)$, we obtain
 \begin{equation}
   \label{eq:m-gamma-case-in-spectrum}
   \mathfrak{m}(\gamma)=\theta^2-\alpha^{-1}(\theta^2-1)\,,
 \end{equation}
where $\alpha$ is the normalizing constant corresponding to the
eigenvalue $\gamma$.
\end{proof}
Suppose now that we are required to enumerate the sequences
$\sigma(J)$ and $\sigma(\widetilde{J})$ according to
Remark~\ref{rem:special-interlace}, but no information is given about
the constant $\gamma$ other than it is not in $\sigma(J)$. Clearly, one does not need this number for accomplishing
this task, as is stated in the following remark.
\begin{remark}
  \label{rem:special-interlace-no-k}
  Assuming that $J$ has discrete spectrum, let $S=\sigma(J)$,
  $\widetilde{S}=\sigma(\widetilde{J})$ be disjoint, and take any
  $\theta\ne 1$ and $h\in\mathbb{R}$. It follows from Proposition
  \ref{prop:interlacing2} that one can find a set $M$ and functions
  $f:M\to S$, $\widetilde{f}:M\to\widetilde{S}$, with the properties
  given in our convention for enumerating sequences (C1), such that there
  exists a unique $k_0\in M$ for which the following conditions hold
  under the assumption that $\lambda_k=f(k)$ and
  $\mu_k=\widetilde{f}(k)$ for $k\in M$:
  \begin{enumerate}[a)]
  \item
    $\widetilde{S}\cap(\lambda_{k_0-1},\lambda_{k_0})=\emptyset$,
  \item $\lambda_k<\mu_k<\lambda_{k+1}\,,\ \forall k\ge k_0$,
  \item $\lambda_{k-1}<\mu_k<\lambda_k\,,\ \forall k<k_0$,
  \end{enumerate}
if $\theta>1$, and
  \begin{enumerate}[a$'$)]
  \item $\widetilde{S}\cap(\lambda_{k_0-1},\lambda_{k_0})
   =\{\mu_{k_0-1},\mu_{k_0}\}$,
  \item $\lambda_k<\mu_{k+1}<\lambda_{k+1}\,,\ \forall k\ge k_0$,
  \item $\lambda_{k-1}<\mu_{k-1}<\lambda_k\,,\ \forall k< k_0$.
  \end{enumerate}
if $\theta<1$
\end{remark}
Before we state the necessary and sufficient conditions for two
sequences to be the spectra of a Jacobi operator $J$ and its
perturbation $\widetilde{J}$, let us introduce the following
parameterized sequence. Suppose that two sequences $\{\lambda_k\}_{k\in
  M}$ and $\{\mu_k\}_{k\in M}$ are given and enumerated by the set $M$
as convened before. Whenever the series
\begin{equation*}
  \sum_{k\in M}(\mu_k-\lambda_k)
\end{equation*}
converges, the sequence
  \begin{equation}
    \label{eq:tau-def-1}
    \tau_n(\omega):=
   \frac{(\mu_n-\lambda_n)
\displaystyle\prod_{\substack{k\in M\\k\ne n}}
\frac{\lambda_n-\mu_k}{\lambda_n-\lambda_k}
}{(\lambda_n-\omega)\left(\displaystyle\prod_{k\in M}
\frac{\omega-\mu_k}{\omega-\lambda_k}-1\right)}\,,
  \quad \forall n\in M\,.
  \end{equation}
is well defined for any $\omega\in\mathbb{R}$.

\begin{theorem}
  \label{prop:sufficient-1}
  Let $S$ and $\widetilde{S}$ be two disjoint infinite real sequences
  without finite points of accumulation. There exist
  $\theta>1$, $h\in\mathbb{R}$, and a matrix (\ref{eq:jm-0})
  such that $S=\sigma(J)\not\ni\gamma$ and
  $\widetilde{S}=\sigma(\widetilde{J})$ if and only if the following
  conditions hold:
 \begin{enumerate}[i)]
 \item There exist a set $M$ and functions $h:M\to S$,
   $\widetilde{h}:M\to\widetilde{S}$ with the properties given in our
   convention for enumerating sequences (C1) such that one can find a
   unique $k_0\in M$ for which a),b),c) of
   Remark~\ref{rem:special-interlace-no-k} take place with
   $\lambda_k=h(k)$ and $\mu_k=\widetilde{h}(k)$.
 \label{interlace-sufficient}
\item The series $\sum_{k\in M}(\mu_k-\lambda_k)$ is convergent.
\label{convergence-sufficient}
\item There exists $\widehat{\omega}\in (\lambda_{k_0-1},\lambda_{k_0})$ such
  that
  \begin{enumerate}[a)]
\item For $m=0,1,2,\dots$, the series
\begin{equation*}
    \sum_{k\in M}\lambda_k^{2m}\tau_k(\widehat{\omega})
    \quad\text{converges.}
\end{equation*}
\label{finite-moments-sufficient}
\item If a sequence of complex numbers $\{\beta_k\}_{k\in M}$
  is such that the series
  \begin{equation*}
    \sum_{k\in
      M}\abs{\beta_k}^2\tau_k(\widehat{\omega})
\quad\text{converges}
  \end{equation*}
and, for $m=0,1,2,\dots$,
\begin{equation*}
  \sum_{k\in
    M}\beta_k\lambda_k^m\tau_k(\widehat{\omega})=0\,,
\end{equation*}
then $\beta_k=0$ for all $k\in M$.
\label{density-poly-sufficient}
 \end{enumerate}
\end{enumerate}
\end{theorem}
\begin{proof}
  Due to Propositions~\ref{prop:interlacing2} and
  \ref{prop:convergence-eigenvalues2}, for proving the necessity of
  the conditions, it only remains to show the existence of
  $\widehat{\omega}$ in $(\lambda_{k_0-1},\lambda_{k_0})$ such that
  $\tau_n(\widehat{\omega})=\alpha_n^{-1}$
  for all $n\in M$. Indeed
  \emph{\ref{finite-moments-sufficient}}) and
  \emph{\ref{density-poly-sufficient}}) will follow from the fact that
  all moments of the spectral measure (\ref{eq:rho-discrete}) exist
  and that the polynomials are dense in $L_2(\mathbb{R},\rho)$.

  Clearly, $\gamma\in(\lambda_{k_0-1},\lambda_{k_0})$, so let
  $\widehat{\omega}=\gamma$. Then, from
  (\ref{eq:m-discrete}),(\ref{eq:m-through-m2}), and
  Proposition~\ref{prop:m-goth-actual-form2}, it follows that
\begin{equation}
\label{eq:expression-for-normalizing-constants}
\begin{split}
  \alpha_n^{-1}&=\frac{1}{\theta^2-1}\lim_{\zeta\to\lambda_n}
  \frac{\lambda_n-\zeta}{\zeta-\gamma}\mathfrak{m}(\zeta)\\
  &=\frac{\mu_n-\lambda_n}{(\lambda_n-\gamma)(\theta^2-1)}
  \prod_{\substack{k\in M\\k\ne n}}
    \frac{\lambda_n-\mu_k}{\lambda_n-\lambda_k}\,.
  \end{split}
\end{equation}
Hence, taking into account (\ref{eq:n-goth-at-0}), one verifies that
$\tau_n(\widehat{\omega})=\alpha_n^{-1}$.

 We now prove that conditions
  \emph{\ref{interlace-sufficient}}),
  \emph{\ref{convergence-sufficient}}),
  \emph{\ref{finite-moments-sufficient}}), and
  \emph{\ref{density-poly-sufficient}}) are sufficient.

The condition \emph{\ref{interlace-sufficient}}) implies that
\begin{equation*}
  \frac{\lambda_n-\mu_k}{\lambda_n-\lambda_k}>0\qquad
  \forall k\in M\,,\ k\ne n
\end{equation*}
On the other hand, by \emph{\ref{convergence-sufficient}}) one can
define the number
\begin{equation}
  \label{eq:kappa-def}
  \vartheta=+\sqrt{\prod_{k\in M}
\frac{\widehat{\omega}-\mu_k}{\widehat{\omega}-\lambda_k}}
\end{equation}
which is clearly strictly greater than $1$ since if
$\widehat{\omega}\in(\lambda_{k_0-1},\lambda_{k_0})$, then
$\abs{\widehat{\omega}-\mu_k}>\abs{\widehat{\omega}-\lambda_k}$ for
all $k\in M$. Thus,
\begin{equation*}
  \frac{\mu_n-\lambda_n}{(\lambda_n-\widehat{\omega})(\vartheta^2-1)}>0
  \qquad
  \forall n\in M
\end{equation*}
Hence, for all $n\in M$, $\tau_n(\widehat{\omega})>0$, so define the function
\begin{equation}
  \label{eq:rho-fro-proof}
  \rho(t):=\sum_{\lambda_k<t}\tau_k(\widehat{\omega})\,.
\end{equation}
It follows from  \emph{\ref{finite-moments-sufficient}}) that the
moments of the measure corresponding to $\rho$ are finite.

Now, on the basis of \emph{\ref{interlace-sufficient}}) and
  \emph{\ref{convergence-sufficient}}), define the meromorphic functions
\begin{equation*}
  \check{\mathfrak{m}}(\zeta):=
\prod_{k\in M}\frac{\zeta-\mu_k}{\zeta-\lambda_k}
\end{equation*}
and
\begin{equation}
  \label{eq:definition-m-tilde}
  \check{m}(\zeta):=
  \frac{\check{\mathfrak{m}}(\zeta)-
    \vartheta^2}
  {(\zeta-\widehat{\omega})\left(\vartheta^2-1\right)}\,.
\end{equation}
Thus, taking into account (\ref{eq:tau-def-1}), one has
\begin{equation}
  \label{eq:residue-tilde}
  \res_{\zeta=\lambda_n}\check{m}(\zeta)=
\left(\vartheta^2-1\right)^{-1}\lim_{\zeta\to\lambda_n}
\frac{\zeta-\lambda_n}{\zeta-\widehat{\omega}}\check{\mathfrak{m}}(\zeta)
=-\tau_n(\widehat{\omega})\,.
\end{equation}
Therefore, on the basis of the second equality in
(\ref{eq:limits-of-m-goth}),
\begin{equation}
  \label{eq:limit-tilde}
  \lim_{\substack{\zeta\to\infty \\ \im \zeta\ge\epsilon>0}}
    \check{m}(\zeta)=\left(\vartheta^2-1\right)^{-1}
\lim_{\substack{\zeta\to\infty \\ \im \zeta\ge\epsilon>0}}
\frac{\check{\mathfrak{m}}(\zeta)}
{\zeta-\widehat{\omega}}=0
\end{equation}
By (\ref{eq:residue-tilde}) and (\ref{eq:limit-tilde}),
\cite[Chap. VII, Sec.1 Theorem 2]{MR589888} implies that
\begin{equation}
  \label{eq:m-tilde-as-sum}
  \check{m}(\zeta)=
\sum_{k\in M}\frac{\tau_k(\widehat{\omega})}{\lambda_k-\zeta}\,.
\end{equation}
On the other hand, using again the first equality in
(\ref{eq:limits-of-m-goth}), one obtains
\begin{equation*}
  \lim_{\substack{\zeta\to\infty \\ \im \zeta\ge\epsilon>0}}
    \zeta\check{m}(\zeta)=
\left(\vartheta^2-1\right)^{-1}
\lim_{\substack{\zeta\to\infty \\ \im \zeta\ge\epsilon>0}}
\frac{1}{1-\widehat{\omega}/\zeta}
\left(\check{\mathfrak{m}}(\zeta)-
\vartheta^2\right)=-1\,.
\end{equation*}
But
\begin{equation*}
  \lim_{\substack{\zeta\to\infty \\ \im \zeta\ge\epsilon>0}}
    \zeta\check{m}(\zeta)=-\sum_{k\in M}\tau_k(\widehat{\omega})\,,
\end{equation*}
so it has been proven that, for the function given in
(\ref{eq:rho-fro-proof}),
\begin{equation*}
  \int_{\mathbb{R}} d\rho(t)=1\,.
\end{equation*}
Thus the measure corresponding to $\rho$ is appropriately normalized
and, because of \emph{\ref{finite-moments-sufficient}}), all the
moments exist, so in $L_2(\mathbb{R},\rho)$ apply the Gram-Schmidt
procedure of orthonormalization to the sequence $\{t^k\}_{k=0}^\infty$
to obtain a Jacobi matrix as was explained in the
Section~\ref{sec:preliminaries}. Consider the operator $J_0$ with
domain $l_{\rm fin}(\mathbb{N})$ generated by this Jacobi matrix as
explained in the Introduction. Now, as a consequence of condition
\emph{\ref{density-poly-sufficient}}), which means that the
polynomials are dense in $L_2(\mathbb{R},\rho)$, $\rho$ corresponds to the
resolution of the identity of a self-adjoint extension $J$ of $J_0$
\cite[Prop.\,4.15]{MR1627806}.

Finally, consider
\begin{equation}
  \label{eq:perturbed-family}
    \widetilde{J}=J+
  [q_1(\theta^2-1)+\theta^2h]\inner{\delta_1}{\cdot}\delta_1 +
  b_1(\theta-1)(\inner{\delta_1}{\cdot}\delta_2 +
  \inner{\delta_2}{\cdot}\delta_1)\,,
\end{equation}
where
\begin{equation*}
   \theta=\vartheta\,,\qquad
  h=\widehat{\omega}\frac{1-\vartheta^2}
{\vartheta^2}\,.
\end{equation*}
By construction
the sequence $\{\lambda_k\}_{k\in M}$ is the spectrum of
$J$. For the proof to be complete it only remains to show that
$\{\mu_k\}_{k\in M}$ is the spectrum of $\widetilde{J}$. For the
function given in (\ref{eq:m-goth-def2}), taking into account
(\ref{eq:m-discrete}) and (\ref{eq:m-through-m2}), one has
\begin{equation*}
  \mathfrak{m}(\zeta)=\theta^2+(\zeta-\widehat{\omega})\left(\theta^2-1\right)
  \sum_{k\in M}\frac{1}{\alpha_k(\lambda_k-\zeta)}\,.
\end{equation*}
On the other hand, from (\ref{eq:definition-m-tilde}) and
(\ref{eq:m-tilde-as-sum}), it follows that
\begin{equation*}
  \check{\mathfrak{m}}(\zeta)=\vartheta^2+(\zeta-\widehat{\omega})\left(\vartheta^2-1\right)
  \sum_{k\in M}\frac{\tau_k(\widehat{\omega})}{\lambda_k-\zeta}\,.
\end{equation*}
But $\theta=\vartheta$ and we have already proven that
$\alpha_k^{-1}=\tau_k(\widehat{\omega})$ for $k\in M$. Thus
$\mathfrak{m}=\check{\mathfrak{m}}$, meaning that the zeros of
$\mathfrak{m}$ are given by the sequence $\{\mu_k\}_{k\in M}$.
\end{proof}
\begin{remark}
  \label{rem:uniqueness}
  In accordance with Theorem \ref{prop:reconstruction2}, the proof of
  Theorem~\ref{prop:sufficient-1} shows that the sequences $S$,
  $\widetilde{S}$, and the parameter $\widehat{\omega}$ satisfying
  \emph{\ref{interlace-sufficient})},
  \emph{\ref{convergence-sufficient})}, and \emph{iii)}, uniquely
  determine the perturbation parameters $\theta$ and $h$, and the
  matrix (\ref{eq:jm-0}) with the boundary condition at infinity if
  necessary. Thus, $S$, $\widetilde{S}$, and $\widehat{\omega}$ amount
  to the complete input data for solving uniquely the inverse spectral
  problem.
\end{remark}
\begin{remark}
  \label{rem:abc-prime}
  Clearly, the assertion of Theorem~\ref{prop:sufficient-1} holds true
  if one substitutes $\theta>1$ by $\theta<1$, conditions a), b), c)
  by a$'$), b$'$), c$'$), and $\widehat{\omega}\in(\lambda_{k_0-1},\lambda_{k_0})$ by
  $\widehat{\omega}\in(\mu_{k_0-1},\mu_{k_0})$.
\end{remark}
\begin{proposition}
  \label{prop:other-gammas}
  Let $S$ and $\widetilde{S}$ be two infinite real sequences without
  finite points of accumulation that satisfy
  \ref{interlace-sufficient}) and \ref{convergence-sufficient}) of
  Theorem~\ref{prop:sufficient-1}. Suppose that there is
  $\widehat{\omega}\in(\lambda_{k_0-1},\lambda_{k_0})$ so that the
  sequence $\{\tau_n(\widehat{\omega})\}_{n\in M}$ satisfies
  \ref{finite-moments-sufficient}) and \ref{density-poly-sufficient})
  of Theorem~\ref{prop:sufficient-1}, then $\{\tau_n(\omega)\}_{n\in
    M}$ also satisfies \ref{finite-moments-sufficient}) and
  \ref{density-poly-sufficient}) for all
  $\omega\in(\lambda_{k_0-1},\lambda_{k_0})$.
\end{proposition}
\begin{proof}
  Let
  \begin{equation*}
    \rho_\omega(t):=\sum_{\lambda_k<t}\tau_k(\omega)\,.
  \end{equation*}
  As in the proof of Theorem~\ref{prop:sufficient-1} one verifies that
  if $\omega\in(\lambda_{k_0-1},\lambda_{k_0})$, $\rho_\omega$ is a
  positive non-decreasing function and that
  \begin{equation*}
    \int_{\mathbb{R}}
    d\rho_\omega(t)=1\,,\qquad\forall \omega\in(\lambda_{k_0-1},\lambda_{k_0}).
  \end{equation*}
  By hypothesis all the moments of the measure
  $\rho_{\widehat{\omega}}$ are finite and the polynomials are dense
  in $L_2(\mathbb{R},\rho_{\widehat{\omega}})$. For the proposition to be
  proven, one needs to show that this implies that all the moments of
  the measure $\rho_\omega$ are finite and the polynomials are dense
  in $L_2(\mathbb{R},\rho_\omega)$ for all
  $\omega\in(\lambda_{k_0-1},\lambda_{k_0})$. But, since the support
  of the measure is the same for all
  $\omega\in(\lambda_{k_0-1},\lambda_{k_0})$, this implication will
  indeed take place if for any fixed
  $\omega\in(\lambda_{k_0-1},\lambda_{k_0})$ there are positive
  constants $C_1,C_2$ such that
  \begin{equation}
    \label{eq:ineq-to-proof}
    C_1\tau_n(\widehat{\omega})\le\tau_n(\omega)\le
    C_2\tau_n(\widehat{\omega})\,,\quad\forall n\in M\,.
  \end{equation}
  Fix $\omega\in(\lambda_{k_0-1},\lambda_{k_0})$. From
  (\ref{eq:tau-def-1}), it follows that
 \begin{equation}
   \label{eq:relation-gammas}
   \tau_n(\omega)=C
   \frac{\lambda_n-\widehat{\omega}}
   {\lambda_n-\omega}\tau_n(\widehat{\omega})\,,
 \end{equation}
where $C=\abs{\prod_{k\in M}
\frac{\omega-\mu_k}{\omega-\lambda_k}-1}^{-1}\abs{\prod_{k\in M}
\frac{\widehat{\omega}-\mu_k}{\widehat{\omega}-\lambda_k}-1}$. By
elementary estimates of $\abs{\frac{\lambda_n-\widehat{\omega}}
   {\lambda_n-\omega}}$, one verifies from (\ref{eq:relation-gammas})
 that if
 \begin{equation*}
   C_1:=\frac{\min\{\abs{\lambda_{k_0}-\widehat{\omega}},
 \abs{\lambda_{k_0-1}-\widehat{\omega}}\}}{
\max\{\abs{\lambda_{k_0}-\omega},
 \abs{\lambda_{k_0-1}-\omega}\}}\quad
C_2:=1+\frac{\max\{\abs{\lambda_{k_0}-\widehat{\omega}},
 \abs{\lambda_{k_0-1}-\widehat{\omega}}\}}{\min\{\abs{\lambda_{k_0}-\omega},
 \abs{\lambda_{k_0-1}-\omega}\}}\,,
 \end{equation*}
then (\ref{eq:ineq-to-proof}) holds.
\end{proof}
\begin{remark}
  \label{rem:abc-prime2}
  As in Remark~\ref{rem:abc-prime}, the assertion of
  Proposition~\ref{prop:other-gammas} holds true if one assumes that
  i) is satisfied with a$'$), b$'$),c$'$) instead of a), b), c) and
  substitute the interval $(\lambda_{k_0-1},\lambda_{k_0})$ by
  $(\mu_{k_0-1},\mu_{k_0})$.
\end{remark}
\begin{theorem}
  \label{prop:other-matrices-same-spectra}
  Let $\theta\ne 1$ and assume that the disjoint sets $\sigma(J)$ and
  $\sigma(\widetilde{J})$ are enumerated according to
  Remark~\ref{rem:special-interlace-no-k} with a), b), c) if
  $\theta>1$, and with a$'$), b$'$), c$'$) otherwise. Then, for any
  $\omega\in(\lambda_{k_0-1},\lambda_{k_0})$ when $\theta>1$ and for
  any $\omega\in(\mu_{k_0-1},\mu_{k_0})$ when $\theta<1$, there is a
  matrix
\begin{equation}
    \label{eq:jm-other}
  \begin{pmatrix}
    q_1' & b_1' & 0  &  0  &  \cdots
\\[1mm] b_1' & q_2' & b_2' & 0 & \cdots \\[1mm]  0  &  b_2'  & q_3'  &
b_3' &  \\
0 & 0 & b_3' & q_4' & \ddots\\ \vdots & \vdots &  & \ddots
& \ddots
  \end{pmatrix}\,,
\end{equation}
where $q_n'\in\mathbb{R}$ and $b_n'>0$ for all $n\in\mathbb{N}$,
and a self-adjoint
extension $J'$ of the operator whose matrix representation is
(\ref{eq:jm-other}),
such that $\sigma(J')=\sigma(J)$ and
$\sigma(\widetilde{J'})=\sigma(\widetilde{J})$, where
\begin{equation}
\label{eq:j-prime-tilde}
  \widetilde{J'}:=J'+
  [q_1'((\theta')^2-1)+(\theta')^2h']\inner{\delta_1}{\cdot}\delta_1 +
  b_1'(\theta'-1)(\inner{\delta_1}{\cdot}\delta_2 +
  \inner{\delta_2}{\cdot}\delta_1)
\end{equation}
with
\begin{equation*}
  \theta':=+\sqrt{\mathfrak{m}(\omega)}\,,\qquad
  h':=\omega\frac{1-\mathfrak{m}(\omega)}{\mathfrak{m}(\omega)}\,.
\end{equation*}
\end{theorem}
\begin{proof}
  We prove the assertion for $\theta>1$. The other case is completely
  analogous, one only has to take into account
  Remarks~\ref{rem:abc-prime} and \ref{rem:abc-prime2}. By
  Theorem~\ref{prop:sufficient-1}, it follows that $\sigma(J)$ and
  $\sigma(\widetilde{J})$ satisfy i), ii),
  \ref{finite-moments-sufficient}), and
  \ref{density-poly-sufficient}). Then, from
  Proposition~\ref{prop:other-gammas},
  \ref{finite-moments-sufficient}) and \ref{density-poly-sufficient})
  are satisfied for any
  $\omega\in(\lambda_{k_0-1},\lambda_{k_0})$. Now, again by
  Theorem~\ref{prop:sufficient-1}, there are operators $J'$ and
  $\widetilde{J'}$ such that their spectra coincide with $\sigma(J)$
  and $\sigma(\widetilde{J})$.
\end{proof}
\begin{lemma}
  \label{lem:m-theta-solutions}
Let $\theta\ne 1$ and assume that the disjoint sets $\sigma(J)$ and
$\sigma(\widetilde{J})$ are enumerated according to
Remark~\ref{rem:special-interlace-no-k} with a), b), c) if $\theta>1$,
and with a'), b'), c') otherwise. Then, the equation
\begin{equation}
  \label{eq:theta-equation}
  \mathfrak{m}\upharpoonright_{(\lambda_{k_0-1},\lambda_{k_0})}(s)=\theta^2
\end{equation}
has only the solutions $s=\gamma$ and $s=\widehat{\gamma}$, where
$\widehat{\gamma}$ is the only point in $\sigma(J_{\rm
  T})\cap(\lambda_{k_0-1},\lambda_{k_0})$. Moreover, if
$\gamma=\widehat{\gamma}$, then $\gamma$ is a local extremum of
$\mathfrak{m}\upharpoonright_{(\lambda_{k_0-1},\lambda_{k_0})}$.
\end{lemma}
\begin{proof}
  First notice that $\gamma$ is in $(\lambda_{k_0-1},\lambda_{k_0})$
  if $\theta>1$ and in $(\mu{}_{k_0-1},\mu{}_{k_0})$ otherwise.  By
  Proposition~\ref{prop:interlace-truncated} the set $\sigma(J_{\rm
    T})\cap(\lambda_{k_0-1},\lambda_{k_0})$, has only one element. If
  $\theta<1$, since $J_{\rm T}=\widetilde{J}_{\rm T}$, this only
  element is actually in $(\mu_{k_0-1},\mu_{k_0})$. Moreover, when
  $\theta<1$, by what was said in the proof of
  Proposition~\ref{prop:interlacing2}, $m\upharpoonright_{\mathbb{R}}$
  takes negative values outside $(\mu_{k_0-1},\mu_{k_0})$.  Now, from
  (\ref{eq:m-through-m2}), the solutions of (\ref{eq:theta-equation})
  are the zeros of $(\zeta -\gamma)m(\zeta)$ which are $\gamma$ and
  $\widehat{\gamma}$. Clearly, if $\gamma=\widehat{\gamma}$, the
  function $(\zeta -\gamma)m(\zeta)$ has a zero of multiplicity two
  which implies the second assertion.
\end{proof}
\begin{lemma}
  \label{lem:one-extremum}
  Let $\theta\ne 1$ and assume that the disjoint sets $\sigma(J)$ and
  $\sigma(\widetilde{J})$ are enumerated according to
  Remark~\ref{rem:special-interlace-no-k} with a), b), c) if
  $\theta>1$, and with a'), b'), c') otherwise. Then, the function
  $\mathfrak{m}\upharpoonright_{(\lambda_{k_0-1},\lambda_{k_0})}$ has only one
  local extremum in $(\lambda_{k_0-1},\lambda_{k_0})$ when $\theta>1$,
  and in $(\mu_{k_0-1},\mu_{k_0})$ when $\theta<1$, which turns out to
  be a global minimum greater than 1 if $\theta>1$, and a global
  maximum less that 1 if $\theta<1$.
\end{lemma}
\begin{proof}
  Suppose that $\theta>1$ and that
  $\mathfrak{m}\upharpoonright_{(\lambda_{k_0-1},\lambda_{k_0})}$ has
  more than one local extremum. Then one verifies that there are three
  different points $\omega_1,\omega_2,\omega_3$ in
  $(\lambda_{k_0-1},\break\lambda_{k_0})$ such that
  \begin{equation*}
    \mathfrak{m}(\omega_1)=\mathfrak{m}(\omega_2)=
    \mathfrak{m}(\omega_3)\,.
  \end{equation*}
  By Theorem
  \ref{prop:other-matrices-same-spectra}, for $\omega_1$ there are
  Jacobi operators $J'$ and $\widetilde{J'}$ such that
  $\sigma(J')=\sigma(J)$ and
  $\sigma(\widetilde{J'})=\sigma(\widetilde{J})$. Let $\mathfrak{n}$
  be the quotient of the Weyl $m$-function of $J'$ and the Weyl
  $m$-function of $\widetilde{J'}$. By
  Proposition~\ref{prop:m-goth-actual-form2},
  $\mathfrak{m}=\mathfrak{n}$. Hence, on the basis of Theorem
  \ref{prop:other-matrices-same-spectra}, it follows that
  \begin{equation}
    \label{eq:contradiction}
     \mathfrak{n}(\omega_1)=\mathfrak{n}(\omega_2)=
    \mathfrak{n}(\omega_3)=(\theta')^2\,.
  \end{equation}
 On the other hand, Lemma~\ref{lem:m-theta-solutions} tells us that the equation
 \begin{equation*}
     \mathfrak{n}\upharpoonright_{(\lambda_{k_0-1},\lambda_{k_0})}(s)
     =(\theta')^2\,,
 \end{equation*}
where $\theta'=+\sqrt{\mathfrak{m}(\omega_1)}$, has only the solutions
$\omega_1$ and the only element of $\sigma(J_{\rm T}')$ in
$(\lambda_{k_0-1},\lambda_{k_0})$. This is in contradiction with
(\ref{eq:contradiction}).

Thus there is only one extremum of
$\mathfrak{m}\upharpoonright_{(\lambda_{k_0-1},\lambda_{k_0})}$ when
$\theta>1$. The same reasoning given above, but replacing all
appearances of the interval $(\lambda_{k_0-1},\lambda_{k_0})$ by
$(\mu_{k_0-1},\mu_{k_0})$, works for the case $\theta<1$.

Now, on the basis of the behavior of $\mathfrak{m}$ in the
interval $(\lambda_{k_0-1},\lambda_{k_0})$ if $\theta>1$, and in
$(\mu_{k_0-1},\mu_{k_0})$ if $\theta<1$, given in the proof of
Proposition~\ref{prop:interlacing2}, one completes the proof.
\end{proof}
\begin{theorem}
  \label{prop:two-systems-same-theta}
  Under the assumptions of Lemma~\ref{lem:m-theta-solutions}, if
  $\gamma\ne\widehat{\gamma}$, then there are exactly two different
  matrices (\ref{eq:jm-0}) and (\ref{eq:jm-other}) such that
  $\sigma(J')=\sigma(J)$ and
  $\sigma(\widetilde{J'})=\sigma(\widetilde{J})$ with
  $\theta=\theta'$. If $\gamma=\widehat{\gamma}$, then for all
  operators $J'\ne J$ for which $\sigma(J')=\sigma(J)$ and
  $\sigma(\widetilde{J'})=\sigma(\widetilde{J})$ it turns out that
  $\theta\ne\theta'$.
\end{theorem}
\begin{proof}
  Due to Theorem~\ref{prop:other-matrices-same-spectra} and
  Lemmas~\ref{lem:m-theta-solutions} and \ref{lem:one-extremum}, the
  proof is straightforward.
\end{proof}

\begin{remark}
  \label{rem:max-min-last}
  Clearly, the condition $\gamma=\widehat{\gamma}$ is equivalent to
  $\gamma$ being equal to the minimum of
  $\mathfrak{m}\upharpoonright_{(\lambda_{k_0-1},\lambda_{k_0})}$ if $\theta>1$,
  and being equal to the maximum of
  $\mathfrak{m}\upharpoonright_{(\mu_{k_0-1},\mu_{k_0})}$ if $\theta<1$.
\end{remark}

Let us now reformulate and summarize some of our results in terms of
the mass-spring systems mentioned in the Introduction.

Suppose that one knows the spectrum of the Jacobi operator
corresponding to the mass-spring system given in Fig. 1, and then,
after carrying out a mass-spring perturbation on the system as
illustrated in Fig. 2, one is given the new spectrum, which does not
intersect with the first one. Clearly, by the spectra alone, one
determines if $\Delta m$ is positive or negative (see
Proposition~\ref{prop:interlacing2}). For definiteness, suppose that
$\Delta m>0$. If no more information is given, then for any value of
the ratio of masses
$\theta\in(0,\max_{t\in(\mu_{k_0-1},\mu_{k_0})}\mathfrak{m}(t)]$ there
are mass-spring systems corresponding to Figs. 1 and 2 having the
measured spectra (see
Theorem~\ref{prop:other-matrices-same-spectra}). However, when one
knows the ratio of masses $\theta$ then, in general, there are only
two mass-spring systems corresponding to Fig. 1 that comply with the
conditions after the corresponding perturbation (see
Theorem~\ref{prop:two-systems-same-theta}). Moreover, if
\begin{equation*}
  \theta=\max_{t\in(\mu_{k_0-1},\mu_{k_0})}\mathfrak{m}(t)\,,
\end{equation*}
there is only one system with the required properties (see
Theorem~\ref{prop:two-systems-same-theta}).

Let us now turn to the case when $\gamma\in\sigma(J)$ or,
equivalently, when the spectra of $J$ and $\widetilde{J}$ intersect.
Thus, according to Remark~\ref{rem:special-interlace}, consider the
sequences $\{\lambda_k\}_{k\in M}$ and $\{\mu_k\}_{k\in M}$ such that
$\lambda_{k_0}=\mu_{k_0}=\gamma$.  If
\begin{equation*}
  \sum_{k\in M}(\mu_k-\lambda_k)
\end{equation*}
converges, then, for any
$\omega\in\mathbb{R}$ and $n\in M$, one defines
  \begin{equation}
    \label{eq:upsilon-def-1}
    \upsilon_n(\omega):=
\begin{cases}
  \displaystyle
   \frac{(\mu_n-\lambda_n)}{(\lambda_n-\gamma)(\omega-1)}
\prod_{\substack{k\in M\\k\ne n}}
\frac{\lambda_n-\mu_k}{\lambda_n-\lambda_k}\,,
  &  n\ne k_0\\
\displaystyle
(\omega-1)^{-1}\left(\omega-\prod_{\substack{k\in M\\k\ne k_0}}
\frac{\gamma-\mu_k}{\gamma-\lambda_k}\right) & n= k_0
\end{cases}
\end{equation}
\begin{theorem}
  \label{prop:sufficient-2}
  Let $S$ and $\widetilde{S}$ be two infinite real sequences without
  finite points of accumulation such that
  $S\cap\widetilde{S}=\{\gamma\}$. There exist $\theta>1$,
  $h\in\mathbb{R}$, and a matrix (\ref{eq:jm-0}) such that
  $S=\sigma(J)$ and $\widetilde{S}=\sigma(\widetilde{J})$ if and only
  if the following conditions hold:
 \begin{enumerate}[I)]
 \item There exist a set $M$ and functions $h:M\to S$,
   $\widetilde{h}:M\to\widetilde{S}$ with the properties given in
   Remark~\ref{rem:special-interlace} such that
   (\ref{eq:interlacing-theta-big}) holds and there is a $k_0\in M$
   such that $\lambda_{k_0}=\mu_{k_0}=\gamma$.
 \label{interlace-sufficient-2}
\item The series $\sum_{k\in M}(\mu_k-\lambda_k)$ is convergent.
\label{convergence-sufficient-2}
\item There exists $\widehat{\omega}>\displaystyle
\prod_{\substack{k\in M\\ k\ne k_0}}
\frac{\gamma-\mu_k}{\gamma-\lambda_k}$ such
  that
  \begin{enumerate}[a)]
\item For $m=0,1,2,\dots$, the series
\begin{equation*}
    \sum_{k\in M}\lambda_k^{2m}\upsilon_k(\widehat{\omega})
    \quad\text{converges.}
\end{equation*}
\label{finite-moments-sufficient-2}
\item If a sequence of complex numbers $\{\beta_k\}_{k\in M}$
  is such that the series
  \begin{equation*}
    \sum_{k\in
      M}\abs{\beta_k}^2\upsilon_k(\widehat{\omega})
\quad\text{converges}
  \end{equation*}
and, for $m=0,1,2,\dots$,
\begin{equation*}
  \sum_{k\in
    M}\beta_k\lambda_k^m\upsilon_k(\widehat{\omega})=0\,,
\end{equation*}
then $\beta_k=0$ for all $k\in M$.
\label{density-poly-sufficient-2}
 \end{enumerate}
\end{enumerate}
\end{theorem}
\begin{proof}
  For proving the necessity of the conditions, in view of
  Propositions~\ref{prop:interlacing2} and
  \ref{prop:convergence-eigenvalues2}, one only needs to show the
  existence of $\widehat{\omega}$ strictly greater than
  $\mathfrak{m}(\gamma)$ such that
  $\upsilon_k(\widehat{\omega})=\alpha_k^{-1}$ for all $k\in M$. From
  (\ref{eq:m-gamma-case-in-spectrum}) and the properties of the
  normalizing constants, it follows that
  \begin{equation}
    \label{eq:strange-estimate}
    1<\mathfrak{m}(\gamma)<\theta^2\,.
  \end{equation}
  Let $\widehat{\omega}=\theta^2$, then
  (\ref{eq:expression-for-normalizing-constants}) yields
  $\upsilon_k(\widehat{\omega})=\alpha_k^{-1}$ for $k\in M$, $k\ne
  k_0$. Moreover, (\ref{eq:m-gamma-case-in-spectrum}) implies that
  $\upsilon_{k_0}(\widehat{\omega})=\alpha_{k_0}^{-1}$.

  Let us now prove that \emph{\ref{interlace-sufficient-2}}),
  \emph{\ref{convergence-sufficient-2}}),
  \emph{\ref{finite-moments-sufficient-2}}), and
  \emph{\ref{density-poly-sufficient-2}}) are sufficient.

  It follows from (\ref{eq:interlacing-theta-big}) that
  $\abs{\gamma-\mu_k}>\abs{\gamma-\lambda_k}$ for any $k\in
  M\setminus\{k_0\}$. Since $\gamma-\mu_k$ and $\gamma-\lambda_k$ have
  the same sign,
  \begin{equation*}
    \prod_{\substack{k\in M\\k\ne k_0}}
\frac{\gamma-\mu_k}{\gamma-\lambda_k}>1\,.
  \end{equation*}
  Thus $\widehat{\omega}>1$ and
  $\upsilon_{k_0}(\widehat{\omega})>0$. Now fix $n\in\mathbb{N}$, $n\ne
  k_0$. By \emph{\ref{interlace-sufficient-2}}) one has
\begin{equation*}
  \frac{\lambda_n-\mu_k}{\lambda_n-\lambda_k}>0\qquad
  \forall k\in M\,,\ k\ne n\,.
\end{equation*}
Since $\mu_n-\lambda_n$ and $\lambda_n-\gamma$ are positive or
negative simultaneously, we
conclude that
\begin{equation*}
  \upsilon_n(\widehat{\omega})>0\,,\qquad\forall n\in M\,.
\end{equation*}
Define the function
\begin{equation*}
  \rho(t):=\sum_{\lambda_k<t}\upsilon_k(\widehat{\omega})\,.
\end{equation*}
It follows from  \emph{\ref{finite-moments-sufficient-2}}) that all
the moments of the measure corresponding to $\rho$ are finite.

On the basis of \emph{\ref{interlace-sufficient-2}}) and
  \emph{\ref{convergence-sufficient-2}}), define the meromorphic functions
\begin{equation}
\label{eq:definition-m-frak-tilde-2}
  \check{\mathfrak{m}}(\zeta):=
\prod_{\substack{k\in M\\k\ne k_0}}\frac{\zeta-\mu_k}{\zeta-\lambda_k}
\end{equation}
and
\begin{equation}
  \label{eq:definition-m-tilde-2}
  \check{m}(\zeta):=
  \frac{\check{\mathfrak{m}}(\zeta)-
    \widehat{\omega}}
  {(\zeta-\gamma)\left(\widehat{\omega}-1\right)}\,.
\end{equation}
As it was shown in the proof of Theorem~\ref{prop:sufficient-1}, one
verifies that
\begin{equation*}
  \res_{\zeta=\lambda_n}\check{m}(\zeta)=-\upsilon_n(\widehat{\omega})\,,
  \qquad n\ne k_0\,.
\end{equation*}
It is also straightforward to show that
\begin{equation*}
  \res_{\zeta=\gamma}\check{m}(\zeta)=
\frac{\check{\mathfrak{m}}(\gamma)-\widehat{\omega}}{\widehat{\omega}-1}\,.
\end{equation*}
Thus, since the function $\check{m}(\zeta)$ vanishes as
$\zeta\to\infty$ along curves in the upper complex half plane, according to
\cite[Chap. VII, Sec.1 Theorem 2]{MR589888}, one can
write
\begin{equation}
    \label{eq:m-tilde-as-sum-2}
  \check{m}(\zeta)=
\sum_{k\in M}\frac{\upsilon_k(\widehat{\omega})}{\lambda_k-\zeta}\,.
\end{equation}
From (\ref{eq:m-tilde-as-sum-2}) and the fact that
$\lim_{\substack{\zeta\to\infty \\ \im \zeta\ge\epsilon>0}}
    \zeta\check{m}(\zeta)=-1$, it follows that
\begin{equation*}
  \sum_{k\in M}\upsilon_k(\widehat{\omega})=1\quad
  \text{or, equivalently,}\quad\int_{\mathbb{R}} d\rho(t)=1\,.
\end{equation*}
On the other hand, by \emph{\ref{finite-moments-sufficient-2}}), all the
moments of $\rho$ exist. Hence, using the method explained in Section
~\ref{sec:preliminaries}, one obtains a Jacobi matrix and the operator
$J_0$ generated by it (see the Introduction). Condition
\emph{\ref{density-poly-sufficient-2}}) implies that $\rho$ is the
spectral function of a self-adjoint extension $J$ of $J_0$
\cite[Prop.\,4.15]{MR1627806}. Now, consider (\ref{eq:perturbed-family}),
where now
\begin{equation*}
   \theta=+\sqrt{\widehat{\omega}}\,,\qquad
  h=\gamma\left(\frac{1}{\widehat{\omega}}-1\right)\,.
\end{equation*}
By construction
the sequence $\{\lambda_k\}_{k\in M}$ is the spectrum of
$J$. For the proof to be complete it only remains to show that
$\{\mu_k\}_{k\in M}$ is the spectrum of $\widetilde{J}$. For the
function given in (\ref{eq:m-goth-def2}), taking into account
(\ref{eq:m-discrete}) and (\ref{eq:m-through-m2}), one has
\begin{equation*}
  \mathfrak{m}(\zeta)=\theta^2+(\zeta-\gamma)\left(\theta^2-1\right)
  \sum_{k\in M}\frac{1}{\alpha_k(\lambda_k-\zeta)}\,.
\end{equation*}
In view of (\ref{eq:definition-m-tilde-2}) and
(\ref{eq:m-tilde-as-sum-2}), one has
\begin{equation*}
  \check{\mathfrak{m}}(\zeta)=\widehat{\omega}+(\zeta-\gamma)
\left(\widehat{\omega}-1\right)
  \sum_{k\in M}\frac{\upsilon_k(\widehat{\omega})}{\lambda_k-\zeta}\,.
\end{equation*}
But, since $\theta=\widehat{\omega}$ and the fact that
$\alpha_k^{-1}=\upsilon_k(\widehat{\omega})$ for $k\in M$, it follows
that $\mathfrak{m}=\check{\mathfrak{m}}$. In its turn, this means
that the zeros of $\mathfrak{m}$ are given by the sequence
$\{\mu_k\}_{k\in M}$.
\end{proof}
\begin{remark}
  By repeating the reasoning of the proof of
  Theorem~\ref{prop:sufficient-2}, it is straightforward to verify
  that Theorem~\ref{prop:sufficient-2} remains true if one substitutes
  $\theta>1$ by $\theta<1$, (\ref{eq:interlacing-theta-big}) by
  (\ref{eq:interlacing-theta-small}) in
  \emph{\ref{interlace-sufficient-2})}, and
  \begin{equation*}
    \widehat{\omega}>
\prod_{\substack{k\in M\\ k\ne k_0}}
\frac{\gamma-\mu_k}{\gamma-\lambda_k}\quad\text{by}\quad
    \widehat{\omega}<
\prod_{\substack{k\in M\\ k\ne k_0}}
\frac{\gamma-\mu_k}{\gamma-\lambda_k}\,.
  \end{equation*}
\end{remark}
\begin{proposition}
  \label{prop:other-gammas-2}
  Let $S$ and $\widetilde{S}$ be two infinite real sequences without
  finite points of accumulation such that $S\cap
  \widetilde{S}=\{\gamma\}$ and
  \ref{interlace-sufficient-2}) and \ref{convergence-sufficient-2}) of
  Theorem~\ref{prop:sufficient-2} hold. Suppose that there is
  \begin{equation*}
    \widehat{\omega}>
\prod_{\substack{k\in M\\ k\ne k_0}}
\frac{\gamma-\mu_k}{\gamma-\lambda_k}
  \end{equation*}
so that the
  sequence $\{\upsilon_n(\widehat{\omega})\}_{n\in M}$ satisfies
  \ref{finite-moments-sufficient-2}) and
  \ref{density-poly-sufficient-2}) of Theorem~\ref{prop:sufficient-2},
  then $\{\upsilon_n(\omega)\}_{n\in M}$ also satisfies
  \ref{finite-moments-sufficient-2}) and
  \ref{density-poly-sufficient-2}) for all
  \begin{equation*}
    \omega>
\prod_{\substack{k\in M\\ k\ne k_0}}
\frac{\gamma-\mu_k}{\gamma-\lambda_k}
  \end{equation*}
\end{proposition}
\begin{proof}
  For proving the claim one repeats the reasoning of the proof of
  Proposition~\ref{prop:other-gammas}. Here we observe that, for $n\in
  M$, $n\ne k_0$,
  \begin{equation*}
    \upsilon_n(\omega)=C\upsilon_n(\widehat{\omega})\,,
  \end{equation*}
where $C=\frac{\widehat{\omega}-1}{\omega-1}$.
\end{proof}
\begin{remark}
  \label{rem:other-gammas-other-theta}
  If, in  Proposition~\ref{prop:other-gammas-2}, one substitutes (\ref{eq:interlacing-theta-big}) by
  (\ref{eq:interlacing-theta-small}) in
  \emph{\ref{interlace-sufficient-2})} and
  \begin{equation*}
    \widehat{\omega}>
\prod_{\substack{k\in M\\ k\ne k_0}}
\frac{\gamma-\mu_k}{\gamma-\lambda_k}\,,\qquad
    \omega>
\prod_{\substack{k\in M\\ k\ne k_0}}
\frac{\gamma-\mu_k}{\gamma-\lambda_k}
  \end{equation*}
by
  \begin{equation*}
    \widehat{\omega}<
\prod_{\substack{k\in M\\ k\ne k_0}}
\frac{\gamma-\mu_k}{\gamma-\lambda_k}\,,\qquad
    \omega<
\prod_{\substack{k\in M\\ k\ne k_0}}
\frac{\gamma-\mu_k}{\gamma-\lambda_k}\,,
  \end{equation*}
then the new assertion holds true.
\end{remark}
By repeating the proof of
Theorem~\ref{prop:other-matrices-same-spectra} with a minor
modification one arrives at the following theorem.
\begin{theorem}
  \label{prop:other-matrices-same-spectra2}
  Let $\theta\ne 1$ and assume that the intersecting sets $\sigma(J)$ and
  $\sigma(\widetilde{J})$ are enumerated according to
  Remark~\ref{rem:special-interlace} with
  (\ref{eq:interlacing-theta-big}) if $\theta>1$ and
  (\ref{eq:interlacing-theta-small}) if $\theta<1$. Then, for any
  $\omega>\mathfrak{m}(\gamma)$ when $\theta>1$ and for
  any $\omega<\mathfrak{m}(\gamma)$ when $\theta<1$, there is a
  matrix (\ref{eq:jm-other}) and
a self-adjoint
extension $J'$ of the operator whose matrix representation is
(\ref{eq:jm-other}),
such that $\sigma(J')=\sigma(J)$ and
$\sigma(\widetilde{J'})=\sigma(\widetilde{J})$, where $\widetilde{J'}$
is given by (\ref{eq:j-prime-tilde}) with
\begin{equation*}
  \theta':=+\sqrt{\omega}\,,\qquad
  h':=\gamma\left(\frac{1}{\omega}-1\right)\,.
\end{equation*}
\end{theorem}

Let us now comment on the last results in terms of the perturbed
mass-spring systems.

Assume that the spectra of the mass-spring system given in Fig.\,1 and
Fig.\,2 are given and they intersect. By
Proposition~\ref{prop:interlacing2}, these input data determine the
sign of $\Delta m$. Let us suppose that $\Delta m>0$. Due to
Theorem~\ref{prop:other-matrices-same-spectra2}, for any value of the
ratio of masses $\theta<\mathfrak{m}(\gamma)$ there are mass-spring
systems corresponding to Figs. 1 and 2 having the measured
spectra. The knowledge of the ratio of masses completely
determines the mass-spring systems.

We have given above the ratio of masses as a parameter of the system
when the spectra intersect (see Theorems \ref{prop:sufficient-2} and
\ref{prop:other-gammas-2} where $\omega$ and $\widehat{\omega}$ play
the role of the ratio of masses). This is a ``natural'' choice because
the parameter used in the case when the spectra are disjoint, namely
$\gamma$, is now given with the spectra. There is also another choice
for the parameter: the spring constant $h$. Below we briefly discuss
this parameterization where now the role of the spring constant is
played by $\omega$ and $\widehat{\omega}$. We begin by defining
\begin{equation}
    \label{eq:upsilon-def-2}
    \widetilde{\upsilon}_n(\omega):=
\begin{cases}
\displaystyle
   \frac{(\mu_n-\lambda_n)(\omega+\gamma)}{(\gamma-\lambda_n)\omega}
\prod_{\substack{k\in M\\k\ne n}}
\frac{\lambda_n-\mu_k}{\lambda_n-\lambda_k}\,,
  &  n\in M\,,\quad n\ne k_0\\
\displaystyle
\frac{\omega+\gamma}{\omega}\prod_{\substack{k\in M\\k\ne k_0}}
\frac{\gamma-\mu_k}{\gamma-\lambda_k}-\frac{\gamma}{\omega}
 & n= k_0
\end{cases}
\end{equation}

\begin{theorem}
  \label{prop:sufficient-3}
  Let $S$ and $\widetilde{S}$ be two infinite real sequences without
  finite points of accumulation such that
  $S\cap\widetilde{S}=\{\gamma\}$. There exist $\theta>1$,
  $h\in\mathbb{R}$, and a matrix (\ref{eq:jm-0}) such that $S=\sigma(J)$
  and $\widetilde{S}=\sigma(\widetilde{J})$ if and only if the
  conditions \ref{interlace-sufficient-2} and
  \ref{convergence-sufficient-2} of Theorem \ref{prop:sufficient-2}
  hold along with
 \begin{enumerate}[I')]
\setcounter{enumi}{2}
\item There exists a real number $\widehat{\omega}$ satisfying
  \begin{equation*}
    \widehat{\omega}\,\,
    \begin{cases}
      =0 & \text{ if } \gamma=0\\
      <\gamma\left(\displaystyle\prod_{\substack{k\in M\\ k\ne k_0}}
\frac{\gamma-\lambda_k}{\gamma-\mu_k}-1\right) & \text{ if }\gamma>0\\
 >\gamma\left(\displaystyle\prod_{\substack{k\in M\\ k\ne k_0}}
\frac{\gamma-\lambda_k}{\gamma-\mu_k}-1\right) & \text{ if }\gamma<0
    \end{cases}
  \end{equation*}
such
  that
  \begin{enumerate}[a)]
\item For $m=0,1,2,\dots$, the series
\begin{equation*}
    \sum_{k\in M}\lambda_k^{2m}\widetilde{\upsilon}_k(\widehat{\omega})
    \quad\text{converges.}
\end{equation*}
\label{finite-moments-sufficient-3}
\item If a sequence of complex numbers $\{\beta_k\}_{k\in M}$
  is such that the series
  \begin{equation*}
    \sum_{k\in
      M}\abs{\beta_k}^2\widetilde{\upsilon}_k(\widehat{\omega})
\quad\text{converges}
  \end{equation*}
and, for $m=0,1,2,\dots$,
\begin{equation*}
  \sum_{k\in
    M}\beta_k\lambda_k^m\widetilde{\upsilon}_k(\widehat{\omega})=0\,,
\end{equation*}
then $\beta_k=0$ for all $k\in M$.
\label{density-poly-sufficient-3}
 \end{enumerate}
\end{enumerate}
\end{theorem}
\begin{proof}
  The proof is similar to the one of
  Theorem~\ref{prop:sufficient-2} and we restrict ourselves to the
  case when $\gamma>0$. The other cases are proven analogously.
 Thus, for the necessity of the
  conditions to be proven, one only should establish that there is
  \begin{equation*}
    \widehat{\omega}<
\gamma\left(\displaystyle\prod_{\substack{k\in M\\ k\ne k_0}}
\frac{\gamma-\lambda_k}{\gamma-\mu_k}-1\right)
  \end{equation*}
  such that $\widetilde{\upsilon}_k(\widehat{\omega})=\alpha_k^{-1}$
  for all $k\in M$. On the basis of (\ref{eq:gamma-def}) and
  (\ref{eq:strange-estimate}), one has
  \begin{equation*}
    1<\mathfrak{m}(\gamma)<\frac{\gamma}{\gamma+h}\,.
  \end{equation*}
Note that $\gamma+h\ne 0$. Since $\gamma$ and $\gamma+h$ have always
the same sign and we are assuming that $\gamma>0$, the inequality
\begin{equation*}
  h<\gamma\left(\frac{1}{\mathfrak{m}(\gamma)}-1\right)
\end{equation*}
holds. So let $\widehat{\omega}=h$, then
(\ref{eq:expression-for-normalizing-constants}) yields
$\widetilde{\upsilon}_k(\widehat{\omega})=\alpha_k^{-1}$ for $k\in M$, $k\ne
k_0$. Moreover, (\ref{eq:m-gamma-case-in-spectrum}) implies that
$\widetilde{\upsilon}_{k_0}(\widehat{\omega})=\alpha_{k_0}^{-1}$.

 Let us now prove that \emph{\ref{interlace-sufficient-2}}),
  \emph{\ref{convergence-sufficient-2}}),
  \emph{\ref{finite-moments-sufficient-3}}), and
  \emph{\ref{density-poly-sufficient-3}}) are sufficient.
Reasoning as before, one verifies that
\begin{equation*}
  \widetilde{\upsilon}_n(\widehat{\omega})>0\,,\forall n\in M\,.
\end{equation*}
Now, instead of
  (\ref{eq:definition-m-tilde-2}) one defines
  \begin{equation*}
      \check{m}(\zeta):=\frac{(\widehat{\omega}-\gamma)
      \check{\mathfrak{m}}(\zeta)-\gamma}
    {(\zeta-\gamma)(2\gamma-\widehat{\omega})}\,,
  \end{equation*}
where $\check{\mathfrak{m}}$ is given in
(\ref{eq:definition-m-frak-tilde-2}). Then it is shown that
\begin{equation}
  \label{eq:res-upsilon-tilde}
  \res_{\zeta=\lambda_n}\check{m}(\zeta)=
-\widetilde{\upsilon}_n(\widehat{\omega})\quad\forall n\in M\,.
\end{equation}
Having defined
\begin{equation*}
  \rho(t):=\sum_{\lambda_k<t}\widetilde{\upsilon}_k(\widehat{\omega})\,,
\end{equation*}
the asymptotic behavior of $\zeta\check{m}(\zeta)$ and
(\ref{eq:res-upsilon-tilde}) imply that $\int_{\mathbb{R}} d\rho(t)=1$. Furthermore,
by
\emph{\ref{finite-moments-sufficient-3}}), all the moments exist, so
one constructs the operator $J_0$  as was done before and, by
\emph{\ref{density-poly-sufficient-3}}) $\rho$ corresponds to a
self-adjoint extension $J$ of $J_0$. Let us now consider
(\ref{eq:perturbed-family}) with
\begin{equation*}
  \theta=+\sqrt{\frac{\gamma}{\widehat{\omega}+\gamma}}\,,\qquad
    h=\widehat{\omega}\,.
\end{equation*}
Clearly, $\sigma(J)=\{\lambda_k\}_{k\in M}$. Hence it only remains to
show that $\sigma(\widetilde{J})=\{\mu_k\}_{k\in M}$. By
(\ref{eq:m-discrete}) and (\ref{eq:m-through-m2}), one has
\begin{equation*}
  \mathfrak{m}(\zeta)=\theta^2+(\zeta-\gamma)\left(\theta^2-1\right)
  \sum_{k\in M}\frac{1}{\alpha_k(\lambda_k-\zeta)}\,.
\end{equation*}
On the other hand (\ref{eq:definition-m-tilde-2}) and
(\ref{eq:m-tilde-as-sum-2}) imply that
\begin{equation*}
  \check{\mathfrak{m}}(\zeta)=\frac{\gamma}{\widehat{\omega}+\gamma}
+(\gamma-\zeta)
\frac{\widehat{\omega}}{\widehat{\omega}+\gamma}
  \sum_{k\in M}\frac{\widetilde{\upsilon}_k(\widehat{\omega})}{\lambda_k-\zeta}\,.
\end{equation*}
Since $\theta^2=\gamma/(\widehat{\omega}+\gamma)$, we conclude that
$\mathfrak{m}=\check{\mathfrak{m}}$. In its turn, this means that
the zeros of $\mathfrak{m}$ are given by the sequence $\{\mu_k\}_{k\in
  M}$.
\end{proof}
\begin{remark}
  Theorem~\ref{prop:sufficient-3} holds true after substituting
  $\theta>1$ by $\theta<1$ and instead of
  \begin{equation*}
    \widehat{\omega}\,\,
    \begin{cases}
      =0 & \text{ if } \gamma=0\\
      <\gamma\left(\displaystyle\prod_{\substack{k\in M\\ k\ne k_0}}
\frac{\gamma-\lambda_k}{\gamma-\mu_k}-1\right) & \text{ if }\gamma>0\\
 >\gamma\left(\displaystyle\prod_{\substack{k\in M\\ k\ne k_0}}
\frac{\gamma-\lambda_k}{\gamma-\mu_k}-1\right) & \text{ if }\gamma<0
    \end{cases}
  \end{equation*}
one writes
  \begin{equation*}
    \widehat{\omega}\,\,
    \begin{cases}
      =0 & \text{ if } \gamma=0\\
      >\gamma\left(\displaystyle\prod_{\substack{k\in M\\ k\ne k_0}}
\frac{\gamma-\lambda_k}{\gamma-\mu_k}-1\right) & \text{ if }\gamma>0\\
 <\gamma\left(\displaystyle\prod_{\substack{k\in M\\ k\ne k_0}}
\frac{\gamma-\lambda_k}{\gamma-\mu_k}-1\right) & \text{ if }\gamma<0
    \end{cases}
  \end{equation*}
The proof of this claim proceeds in exactly the same way as the proof
of  Theorem~\ref{prop:sufficient-3}.
\end{remark}

Since assertions analogous to Proposition~\ref{prop:other-gammas-2}
and Remark~\ref{rem:other-gammas-other-theta} hold when one considers
the sequence (\ref{eq:upsilon-def-2}) instead of
(\ref{eq:upsilon-def-1}), the proof of the following statement can be
done by repeating, with just minor modifications, the proof of
Theorem~\ref{prop:other-matrices-same-spectra2}.
\begin{theorem}
    \label{prop:other-matrices-same-spectra3}
  Let $\theta\ne 1$ and assume that the intersecting sets $\sigma(J)$ and
  $\sigma(\widetilde{J})$ are enumerated according to
  Remark~\ref{rem:special-interlace} with
  (\ref{eq:interlacing-theta-big}) if $\theta>1$ and
  (\ref{eq:interlacing-theta-small}) if $\theta<1$. Assume that
  $\gamma>0$, then, for any
  \begin{equation*}
    \omega<\gamma\left(\prod_{\substack{k\in M\\ k\ne k_0}}
\frac{\gamma-\lambda_k}{\gamma-\mu_k}-1\right)
  \end{equation*}
when $\theta>1$, and for any
\begin{equation*}
  \omega>\gamma\left(\prod_{\substack{k\in M\\ k\ne k_0}}
\frac{\gamma-\lambda_k}{\gamma-\mu_k}-1\right)
\end{equation*}
 when $\theta<1$, there is a
  matrix (\ref{eq:jm-other}) and
a self-adjoint
extension $J'$ of the operator whose matrix representation is
(\ref{eq:jm-other}) such that $\sigma(J')=\sigma(J)$ and
$\sigma(\widetilde{J'})=\sigma(\widetilde{J})$, where $\widetilde{J'}$
is given by (\ref{eq:j-prime-tilde}) with
\begin{equation*}
  \theta':=+\sqrt{\frac{\gamma}{\omega+\gamma}}\,,\qquad
  h':=\omega\,.
\end{equation*}
\end{theorem}

\end{document}